\documentclass[11pt]{article}
\usepackage[utf8]{inputenc}
\usepackage[T1]{fontenc}
\usepackage{times}
\usepackage{fixltx2e}
\usepackage{array, multicol, multirow, comment}
\usepackage{graphicx, float, subcaption}
\usepackage{caption}
\usepackage{longtable}
\usepackage{setspace}
\usepackage{wrapfig}
\usepackage{soul}
\usepackage{pdfpages}
\usepackage{amsmath}
\usepackage{mathtools}
\usepackage{textcomp}
\usepackage{marvosym}
\usepackage{wasysym}
\usepackage{latexsym}
\usepackage{amssymb}
\usepackage{longtable}
\usepackage{array}
\usepackage{algorithm}
\usepackage{algpseudocode}
\usepackage{xcolor}
\usepackage[round, colon]{natbib}
\usepackage[margin=1in]{geometry}
\usepackage{enumerate,enumitem}
\usepackage{enumitem}
\usepackage{bigstrut}

\usepackage{lineno}
\usepackage[affil-it]{authblk}
\usepackage[colorlinks=true, linkcolor=blue, citecolor=blue, urlcolor=blue,pdfborder={0 0 0 0}]{hyperref}
\def\T{{ \mathrm{\scriptscriptstyle T} }}

\input{Definitions}
\begin{document} 
\title{Bayesian Compressed Mixed-Effects Models}
\author{Sreya Sarkar \footnote{PhD Candidate, Department of Statistics and Actuarial Science, The University of Iowa}, Kshitij Khare \footnote{Professor, Department of Statistics, University of Florida}, Sanvesh Srivastava \footnote{Associate Professor, Department of Statistics and Actuarial Science, The University of Iowa}}
\date{}

\maketitle

\begin{abstract}
    Penalized likelihood and quasi-likelihood methods dominate inference in high-dimensional linear mixed-effects models. Sampling-based Bayesian inference is less explored due to the computational bottlenecks introduced by the random effects covariance matrix. To address this gap, we propose the compressed mixed-effects (CME) model, which defines a quasi-likelihood using low-dimensional covariance parameters obtained via random projections of the random effects covariance. This dimension reduction, combined with a global-local shrinkage prior on the fixed effects, yields an efficient collapsed Gibbs sampler for prediction and fixed effects selection. Theoretically, when the compression dimension grows slowly relative to the number of fixed effects and observations, the Bayes risk for prediction is asymptotically negligible, ensuring accurate prediction using the CME model.  Empirically, the CME model outperforms existing approaches in terms of predictive accuracy, interval coverage, and fixed-effects selection across varied simulation settings and a real-world dataset.
\end{abstract}

Keywords: Gibbs sampling; mixed model; parameter expansion; quasi-likelihood; random projection

\section{Introduction}
\label{Section:low-dim-model}
\subsection{Setup}\label{Section:motivation}
Linear mixed-effects models are flexible extensions of the linear model for analyzing clustered and repeated-measures data. In addition to the mean and error variance parameters, these models have a covariance parameter to account for dependencies between responses from the same subject. 
The joint estimation of mean and covariance parameters in their high-dimensional extensions poses significant inferential challenges. While penalized likelihood or maximum a posteriori estimation methods are commonly employed, sampling-based Bayesian inference has been largely restricted to low-dimensional settings. The CME model fills this gap by leveraging random projection matrices and shrinkage priors for efficient fixed effects selection and prediction in high-dimensional applications using an efficient collapsed Gibbs sampling algorithm.

Consider the setup of a linear mixed-effects model. Let $n$ be the number of subjects, $m_i$ be the number of observations on subject $i$, $N = \sum_{i=1}^n m_i$ be the total sample size, $y_i \in \mathbb{R} ^ {m_i}$ be the response for subject $i$, and $X_i \in \mathbb{R} ^ {m_i \times p}, Z_i \in \mathbb{R} ^ {m_i \times q}$ be the fixed and random effects design matrices for subject $i$, respectively. Then, we model the response of subject $i$ as
\begin{equation}
\label{eq:mem1}
    y_i = X_i \beta + Z_i b_i + \epsilon_i, \quad b_i \sim \mathcal{N}(0, \tau ^ 2 \Sigma), \quad \epsilon_i \sim \mathcal{N}(0, \tau ^ 2 I_{m_i}), \quad \epsilon_i \perp b_i, \quad i = 1, \dots, n,
\end{equation}
where $\beta \in \RR^{p}$ is the fixed effects parameter shared by the population, $b_i \in \RR^q$ is the random effects vector for subject $i$ that models the dependency in $y_i$, $\tau^2 \Sigma$ is the random effects covariance matrix, $\epsilon_i$ is the idiosyncratic error vector for subject $i$ with variance $\tau^2$, $b_i$ and $\epsilon_i$ are mutually independent, $\mathcal{N}$ denotes the Gaussian distribution, and $I_{m_i}$ is an $m_i \times m_i$ identity matrix. The model parameter is 
$\{\beta, \Sigma, \tau^2\}$ and $p$ or $q$ exceed or are comparable to $N$ in  high-dimensional applications. 

\subsection{Prior Literature on Mixed Models}\label{Section:prior-lit-lmm}

There is extensive literature on mixed models when $p, q, m_1, \ldots, m_n$ are fixed as $N$ tends to infinity \citep{Batetal13}. Restricted maximum likelihood estimators of $\beta, \Sigma, \tau^2$ are  often preferred  because they are less biased, but expectation maximization (\textsc{EM}) algorithms are available for maximum likelihood estimation. These methods assume that $p, q \ll N$ and rely on the Gaussian assumptions in \eqref{eq:mem1} for analytic updates. A misspecified random effects distribution does not bias $\beta$ estimates, but they are inapplicable when $\Sigma$ is singular or either $p$ or $q$ exceeds $N$ \citep{JiaNgu07}. 
 
Penalized likelihood approaches are ideal for these situations. For selecting the fixed effects, a lasso-type penalty on $\beta$ is the primary choice \citep{2014_GrollTutz,2023_Oliveira_etal}. \citet{Bondell2010} jointly select the fixed and random effects in linear mixed models using an adaptive lasso penalty on $\beta$ and Cholesky factor of $\Sigma$. This approach has been extended to generalized linear mixed models \citep{2008_Ibrahim_Zhu_Tang,2011_Ibrahim_etal,2024_glmmPenFA}. Substituting the likelihood with a quasi-likelihood that replaces $\Sigma^{-1}$ with proxy matrices, such as the Moore-Penrose generalized inverse, improves computational efficiency while maintaining theoretical guarantees \citep{2014_Pan_Huang,Hui2017_RegularizedPOL}. Alternatives to these approaches include a two-stage penalized restricted log-likelihood and moment-based estimators \citep{PenLu12,Ahnetal12, Linetal13}. For convergence guarantees, these methods assume that $p$ and $q$ remain fixed as $N$ increases and the optimal performance relies on the proper choice of tuning parameter.

Bayesian methods are the stochastic counterparts of their penalized likelihood versions \citep{Fahetal10,KimKim11}. \citet{CheDun03} extend the spike-and-slab prior to the diagonal entries in the Cholesky factor of $\Sigma$ for selecting random effects. These ideas have been extended for joint fixed and random effects selection in generalized linear mixed models, and efficient approximations of the posterior distributions of $\beta$ and $\Sigma$ are available \citep{KinDun07,Fonetal10,Zho13,2024_Griffin_fastbayeslmm}. Despite these developments, sampling-based Bayesian inference remains under-explored in high-dimensional applications due to the computational bottlenecks faced by Monte Carlo algorithms as $p$ or $q$ increase with $N$.

Penalized likelihood approaches have dominated recent advancements in high-dimensional linear mixed models. When $m_i$'s are finite, fixed effects selection using a lasso penalty on $\beta$ achieves optimal convergence guarantees and benefits from a fast estimation algorithm  \citep{Schetal11, 2014_SchellMB}. Quasi-likelihood methods substitute $\Sigma$ in \eqref{eq:mem1} by a proxy matrix. For fixed and random effects selection, the quasi-likelihood defines the loss, and  the rows of $b_i$'s and $\beta$ are regularized using lasso-type penalties \citep{FanLi12}. \citet{2020_Bradetal} develop a statistic for testing the marginal significance of the fixed effects in \eqref{eq:mem1}. Their test is robust to misspecification of the random effects distribution but requires $q$ and $m_i$'s to be finite. Finally, \cite{Lietal21} introduce another quasi-likelihood, estimate $\beta$ using the lasso penalty, and debias it for inference. They also develop methods for inference on $\Sigma$ but under the condition that $q < \min_i m_i$, which is restrictive in practice. Our numerical results show that the optimal performance of these methods depends heavily on the choice of tuning parameters and the proxy matrices in the quasi-likelihoods. 

\subsection{Prior Literature on Compressed Regression}\label{Section:prior-lit-compreg}

Compressed methods, which use random matrices for dimension reduction, focus on models without random effects, where $b_i = 0$ and $m_i=1$ for $i=1, \ldots, n$ and $N=n$ in \eqref{eq:mem1}. This implies that $X_i \in \RR^p$ and compression projects $X_i$ into a $k$-dimensional space as $R X_i$, where $R$ is a $k \times p$ random matrix, $k \ll \min(n, p)$, and the entries in $R$ are independently drawn from $\Ncal(0, 1/k)$. If $k \ll n$, then the least squares estimator using $RX_i$'s as covariates has optimal predictive performance because the inner product in the compressed space accurately approximates the mean response \citep{2009_Maillard_Munos,Kab14}.  Averaging predictions from several random projection matrices improves the prediction risk \citep{Sla18}.

Bayesian equivalents of these methods sample from the posterior predictive distribution by performing computations in the compressed space. The compressed predictive posterior distribution has a near-optimal parametric contraction rate in high-dimensional settings \citep{GuhDun15,MukDun20}. These methods are distinct from the sketching approaches that reduce the sample size in large $n$ settings for privacy or scalable posterior computations  \citep{Zhoetal09,GuhSch21,Guhetal23}.

\subsection{Our Contributions}\label{Section:contribution}

The CME model defines the quasi-likelihood using a proxy covariance matrix based on two random  matrices. Let $L$ be the Cholesky factor of $\Sigma$ in \eqref{eq:mem1}. Then, CME model substitutes $L$ by $S^\T \Gamma R$ in \eqref{eq:mem1}, where $\Gamma  \in \RR^{k_1 \times k_2}$ is a low-dimensional parameter, and $S \in \RR^{k_1 \times q}$ and $R \in \RR^{k_2 \times q}$ are random matrices with independent $\Ncal(0, 1/k_1)$ and $\Ncal(0, 1 / k_2)$ entries. This reduces the parameter dimension from $O(p+q^2)$  to $O(p+k_1k_2)$. Using the Gaussian likelihood, we assign shrinkage, Gaussian, and inverse-gamma priors on $\beta, \Gamma$, and $\tau^2$ that enable posterior sampling through a collapsed Gibbs sampler. Finally, 
our theoretical results show that if $k_1$ and $k_2$ are small relative to $p, q$, and $N$, then CME model's prediction risk decays to zero, resulting in efficient posterior computations; see Section \ref{Section:Theory} for details.

Our contributions are threefold. First, the CME model's quasi-likelihood uses a proxy matrix for $\Sigma$ that depends on $\Gamma$. At each iteration, $\Gamma$ is sampled from its full conditional distribution, eliminating any need for asymptotic constraints on $\Gamma$. The CME model's regularity assumptions require that $\Gamma$'s singular values are bounded away from zero, which is satisfied by choosing $k_1$ and $k_2$ to be small. Second, the CME model quantifies uncertainty through credible intervals using draws from the posterior distribution of $\beta$ and the posterior predictive distribution. The compression of $L$ to $\Gamma$ via  $R$ and $S$ is the key idea that allows the Gibbs sampler to bypass computational bottlenecks, particularly when $q \gg N$. Finally, CME model's performance is robust to the choice of $k_1$ and $k_2$ across a range of simulated and benchmark data analyses.
 
\section{Bayesian Compressed Mixed-Effects Models}\label{Section:Methodology}

\subsection{Model}\label{Section:model}

Consider the CME model definition. Let $R \in \RR^{k_2 \times q}$ and $S \in \RR^{k_1 \times q}$ be random matrices with entries independently sampled from $\mathcal{N}(0, 1/k_2)$ and $\mathcal{N}(0, 1/k_1)$, respectively, where $k_1 \ll q, k_2 \ll q$. The CME model replaces the lower triangular Choleksy factor $L$ of $\Sigma$ in \eqref{eq:mem1} with the compressed form $S^\top \Gamma R$: 
\begin{equation}\label{eq:compmem-varsel-eq1}
  y_i = X_i \beta + Z_i S^{\T} \Gamma d_i + \epsilon_i, \quad d_i \sim \mathcal{N}(0, \tau^2 RR^{\T}), \quad \epsilon_i \sim \mathcal{N}(0, \tau^2 I_{m_i}), \quad \epsilon_i \perp d_i, \quad i = 1, \dots, n, 
\end{equation}
where $\beta \in \mathbb{R}^p$ is the high-dimensional fixed effects parameter, $\Gamma \in \RR^{k_1 \times k_2}$ is the compressed covariance parameter, $d_i \in \mathbb{R}^{k_2}$ is the compressed random effect independent of $\epsilon_i$. The CME model defines the proxy matrix for $\Sigma$ as $S^{\T}\Gamma R R^{\T} \Gamma^{\T}S$ and reduces the covariance parameter dimension from $O(q^2)$ to $O(k_1 k_2)$. By constraining $\Gamma$ to be a low-dimensional parameter with $\max (k_1, k_2) \ll q$, the CME model ensures the computational efficiency of its Gibbs sampler in high-dimensional applications where $q \gg N$. 

The CME model specification is completed by assigning priors on $\beta, \tau^2$, and $\Gamma$. We assume that $(\beta, \tau^2)$ and $\Gamma$ are a priori independent, with $\Gamma$ being low-dimensional. A Gaussian prior $\mathcal{N}(0, \sigma^2_{\gamma} I_{k_1 k_2})$ is assigned on $\gamma$, where $\gamma $ is the column-wise vectorization of $\Gamma$ and $\sigma^2_{\gamma}$ is a hyperparameter. We choose the Horseshoe prior on $\beta$ due to its optimal variable selection properties \citep{CarPolSco10,Vanetal17,2022_BhatKhare}. The hierarchical specification of the prior on $(\beta, \tau^2)$ is as follows:
\begin{align}\label{gbt-prior}
   \beta_j \mid \lambda_j^2, \delta^2, \tau^2 & \sim \mathcal{N}(0, \lambda_j^2 \delta^2 \tau^2 ), \quad
    \lambda_j \sim \mathcal{C}^{+}(0, 1), \quad \delta \sim \mathcal{C}^{+}(0, 1), \quad  \tau^2 \sim \mathcal{IG}(a_0, b_0),
\end{align}
independently for $j = 1, \ldots, p$, where $\mathcal{C}^{+}(0, 1)$ is the standard half-Cauchy distribution and $\mathcal{IG}(a_0, b_0)$ is the inverse Gamma distribution with shape and scale parameters $a_0$ and $b_0$. The global shrinkage parameter controlling the overall sparsity in $\beta$ is $\delta$, and $\lambda_j$ is the local shrinkage parameter for $\beta_j$.

\subsection{Fixed Effects Selection and Prediction}\label{Section:FES}

The collapsed Gibbs sampler for inference and predictions in the CME model treats $d_1, \ldots, d_n$ as missing data and imputes them in every iteration. The joint distribution of $(y_i, d_i)$ implied by \eqref{eq:compmem-varsel-eq1} is
\begin{align}\label{eq:jt-y-d}
\Ncal_{m_i + k_2} 
    \left\{  \begin{pmatrix}
        X_i \beta \\ 0
    \end{pmatrix} , \tau^2 \begin{pmatrix}
        Z_i S^{\T} \Gamma R R^{\T} \Gamma^{\T} S Z_i^{\T} + I_{m_i} & Z_i S^{\T} \Gamma R R ^ {\T} \\
                         R R ^ {\T} \Gamma^\T S Z_i^\T &  R R ^ {\T}
                    \end{pmatrix}   \right\}.
\end{align}
For $i=1, \ldots, n$,  \eqref{eq:jt-y-d} implies that the marginal covariance of $y_i$ uses $S^{\T} \Gamma R R^{\T} \Gamma^{\T} S$ as a proxy matrix for $\Sigma$, characterizing the quasi-likelihood for fixed effects selection. The full conditional distribution of $d_i$ given $(y_i, \beta, \tau^2, \gamma)$ implied by  \eqref{eq:jt-y-d} is
\begin{align}\label{eq:cond-d}
  \Ncal  ( \mu_{d_i \mid y_i}, V_{d_i \mid y_i}), \quad \mu_{d_i \mid y_i}  = V_{y_i, d_i}^{\T} V_{y_i, y_i}^{-1} \left( y_i - X_i \beta \right), \quad
    V_{d_i \mid y_i} = V_{d_i, d_i} - V_{y_i, d_i}^{\T} V_{y_i, y_i}^{-1} V_{y_i, d_i},
\end{align}
where $V_{y_i, y_i}, V_{d_i, d_i}$ are the diagonal blocks and $V_{y_i, d_i} = \tau^2 Z_i S^{\T} \Gamma R R ^ {\T}$ is the off-diagonal block of the covariance matrix in \eqref{eq:jt-y-d}. 

Given $d_1, \dots, d_n$, the CME model's quasi-likelihood is Gaussian, enabling a systematic scan Gibbs sampler to sequentially draw $(d_1, \dots, d_n, \gamma, \beta, \tau^2)$ in each iteration; however, every iteration requires augmenting $d_1, \ldots, d_n$ and sampling the high-dimensional parameter $\beta$, leading to high serial autocorrelation in the parameter draws and slow convergence to the stationary distribution. To bypass this problem, the collapsed Gibbs sampler first samples the low-dimensional parameter $\gamma$ from its full conditional distribution and then samples $(\beta, \tau^2)$ after marginalizing over the $d_i$'s, resulting in faster convergence than the systematic scan Gibbs sampler. 

Consider sampling $\gamma$ given $(\beta, \tau^2)$, $d_i$'s, and the observed data. We stack $y_i$'s and $\epsilon_i$'s to obtain the column vectors $y \in \RR^N$ and $\epsilon \in \RR^N$. Define $\check Z \in \RR^{N \times k_1 k_2}$ such that $[d_i^{\T} \otimes Z_i S^{\T}]$ is its $i$th row. Given the $d_i$'s and $(\beta, \tau^2)$, \eqref{eq:compmem-varsel-eq1} reduces to a linear regression model with regression coefficient $\gamma$:
\begin{equation}\label{eq:gen-lik-2}
    y - X\beta = \check Z \gamma + \epsilon, \quad \epsilon \sim \mathcal{N}(0, \tau^2 I_N), 
    \quad 
    \gamma \sim \mathcal{N}(0, \sigma^2_{\gamma} I_{k_1 k_2}).    
\end{equation}
The Gaussian likelihood and the prior on $\gamma$ in \eqref{eq:gen-lik-2} imply that the full conditional distribution of $\gamma$ is
\begin{equation}\label{eq:post-gamma}
     \mathcal{N}(\mu_{\gamma}, \Sigma_{\gamma}), \quad \Sigma_{\gamma} = \left( \frac{1}{\tau^{2}} \check{Z}^{\T} \check{Z} + \frac{1}{\sigma_{\gamma}^2} I_{k_1 k_2} \right)^{-1}, \quad \mu_{\gamma} = \frac{1}{\tau^{2}}\Sigma_{\gamma} \check{Z}^{\T} (y - X\beta).
\end{equation}
The computational complexity of sampling $\gamma$ is $O(k_1^3 k_2^3)$, which is efficient in high-dimensional settings.

Finally, we sample $(\beta, \tau^2)$ from its full conditional given $\gamma$ and the observed data. Marginalizing over the $d_i$'s in \eqref{eq:compmem-varsel-eq1} gives 
\begin{equation}\label{eq:comp-mem-cycle2}
   y_i = X_i \beta + \epsilon'_i, \quad \epsilon'_i \sim \mathcal{N}(0, \tau^2 C_i), \quad C_i = Z_i S^{\T} \Gamma R R^{\T} \Gamma^{\T} S Z_i^{\T}+ I_{m_i}, \quad i = 1, \dots, n,
\end{equation}
which is a weighted linear regression model with parameter $(\beta, \tau^2)$. We reduce \eqref{eq:comp-mem-cycle2} to a Bayesian variable selection problem in homoscedastic setting by scaling $y_i$ and $X_i$ as $y_i^* = C_i^{-1/2}y_i$ and $X_i^* = C_i^{-1/2}X_i$, where $C_i^{-1/2}$ is a square root of $C_i^{-1}$. Row-wise stacking $y_i^{*}$ and $X_i^{*}$ across all $i$ gives $y^{*}\in\mathbb{R}^{N}$ and $X^{*}\in\mathbb{R}^{N\times p}$, which gives the regression model 
\begin{equation}\label{eq:model-cycle2}
    y^* = X^* \beta + \epsilon^*, \quad \epsilon^* \sim \mathcal{N}(0, \tau^2 I_N).
\end{equation}
The prior on $(\beta, \tau^2)$ in \eqref{gbt-prior} yields tractable full conditional distributions for all the parameters and auxiliary variables involved in Horseshoe prior, finishing the sampling of $(\beta, \tau^2)$. The derivation of the full conditionals in the Gibbs sampler is provided in the supplementary material. 

Algorithm \ref{algo:FES} summarizes the analytic forms of the full conditionals used in the collapsed Gibbs sampler for selecting the fixed effects and prediction. We use the marginal chain $(\beta^{(t)})_{t=1}^{\infty}$, where $\beta^{(t)}$ is the $\beta$ draw from the $t$th iteration of the sampler. For any $j \in \{1, \ldots, p\}$, credible intervals for $\beta_j$ are obtained using the quantiles of $(\beta_j^{(t)})_{t=1}^{\infty}$, providing uncertainty quantification at any nominal level; however, directly using these intervals for selecting fixed effects often performs poorly in high-dimensional settings. Instead, we adopt the sequential 2-means ($S_2M$) algorithm for selecting the non-zero $\beta_j$'s \citep{LiPati17}. This algorithm uses $(\beta^{(t)}_j)_{t=1}^{\infty}$ $(j=1, \ldots, p)$ to construct two clusters of indices: one for the zero $\beta_j$'s (noises) and the other for the non-zero $\beta_j$'s (signals).

The marginal parameter chain $(\beta^{(t)}, \tau^{2(t)}, \gamma^{(t)})_{t=1}^{\infty}$ in Algorithm \ref{algo:FES} is used for drawing from the posterior predictive distribution. Let $(X_{*}, Z_{*})$ be the fixed and random effects covariates for a ``test'' sample. Then, the test response $y_*^{(t)}$  is drawn as 
\begin{equation} \label{eq:post-pred}
  y_{*}^{(t)} \sim \mathcal{N}(X_{*} \beta^{(t)}, V_{*}^{(t)}), \quad
  V_{*}^{(t)} = \tau^{2(t)} (Z_{*} S^{\T} \Gamma ^ {(t)} R R^{\T} \Gamma ^ {(t)\T} S Z_{*}^{\T} + I), \quad t = 1, 2, \ldots, \infty.
\end{equation}
The chain $(y_{*}^{(t)})_{t=1}^{\infty}$  are draws from the posterior predictive distribution for the test sample. The Monte Carlo average of these draws serves as the estimate of the test response. This approach naturally extends to multiple test samples by independently drawing for each test sample using \eqref{eq:post-pred}.

\begin{algorithm}[h!]
\caption{\textbf{Collapsed Gibbs sampler for Fixed Effects Selection and Prediction}}
\label{algo:FES}
Let $(\beta^{(0)}, \tau^{2(0)}, \gamma^{(0)})$ be the initial parameter values and superscript $(t)$ denote the parameter draw at the $t$th iteration. For $t = 1, 2, \ldots, \infty$, the $t$th iteration of the Gibbs sampler cycles through the following steps:
\begin{enumerate}
    \item Given $(\beta^{(t-1)}, \tau^{2(t-1)}, \gamma^{(t-1)})$, independently draw $d_i^{(t)} \sim \mathcal{N}(\mu_{d_i \mid y_i}^{(t-1)},V_{d_i \mid y_i}^{(t-1)})$ for $i = 1, \dots, n.$
    \item Given $(d_1^{(t)}, \dots, d_n^{(t)})$ and $(\beta^{(t-1)}, \tau^{2(t-1)})$, draw $\gamma^{(t)} \sim \mathcal{N}(\mu_{\gamma} ^ {(t - 1)}, \Sigma_{\gamma} ^ {(t - 1)})$, where $\mu_{\gamma} ^ {(t - 1)}$ and $\Sigma_{\gamma} ^ {(t - 1)}$ are $\mu_{\gamma}$ and $\Sigma_{\gamma}$ in \eqref{eq:post-gamma} with $(\beta, \tau^2) = (\beta^{(t-1)}, \tau^{2(t-1)})$.
    \item Given $(\beta^{(t-1)}, \tau^{2(t-1)}, \gamma^{(t)}, \lambda_1^{2(t-1)}, \dots, \lambda_p^{2(t-1)}, \delta^{2(t-1)}, \nu_1^{(t-1)}, \dots, \nu_p^{(t-1)}, \xi^{(t-1)})$, transform the likelihood in \eqref{eq:comp-mem-cycle2} to \eqref{eq:model-cycle2} and draw
    \begin{enumerate}
        \item $\delta^{2(t)} \sim \mathcal{IG} \left( \frac{p+1}{2}, \frac{1}{\xi^{(t-1)}} + \frac{1}{2\tau^{2(t-1)}}\sum_{j=1}^p \frac{\beta_j^{2(t-1)}}{\lambda_j^{2(t-1)}} \right)$, 
        \item $\xi^{(t)} \sim \mathcal{IG} \left( 1, 1 + \frac{1}{\delta^{2(t)}} \right)$, 
        \item $\lambda_j^{2(t)} \sim \mathcal{IG} \left( 1, \frac{1}{\nu_j^{(t-1)}} + \frac{\beta_j^{2(t-1)}}{2 \delta^{2(t)} \tau^{2(t-1)}} \right)$,
        \item $\nu_j^{(t)} \sim \mathcal{IG} \left(1, 1 + \frac{1}{\lambda_j^{2(t)}} \right)$,
        \item $\tau^{2(t)} \sim \mathcal{IG} \left( a_0 + \frac{N+p}{2}, b_0 + \frac{(y^{*(t)} - X^{*(t)} \beta^{(t-1)})^{\T}(y^{*(t)} - X^{*(t)} \beta^{(t-1)})}{2} + \frac{\beta^{(t-1)\T}(\delta^{2(t)}\Lambda^{(t)})^{-1}\beta^{(t-1)}}{2} \right)$,
        \item $\beta^{(t)} \sim \mathcal{N}(A^{(t)}X^{*(t) \T}y^{*(t)}, \tau^{2(t)} A^{(t)})$, where $ A^{(t)} = \{ X^{*(t)\T}X^{*(t)} + (\delta^{2(t)} \Lambda^{(t)})^{-1}\}^{-1}$ and $\Lambda^{(t)} = \text{diag}(\lambda^{2(t)}_1, \dots, \lambda^{2(t)}_p)$.
    \end{enumerate}
    \item \textbf{Posterior Predictive Sampling:} For a test sample with fixed and random effect covariates $(X_*, Z_*)$, draw the test response $y_{*}^{(t)} \sim \mathcal{N}(X_{*} \beta^{(t)}, V_{*}^{(t)})$, where $V_{*}^{(t)}$ is defined in \eqref{eq:post-pred}.
\end{enumerate}
\end{algorithm}

\section{Theoretical Properties}\label{Section:Theory}

In this section, we investigate the decay of the posterior prediction risk corresponding to the CME model in a high-dimensional asymptotic setting. We consider a framework where the marginalized (true) data-generating model and (working) compressed model for the $i$th subject ($i=1, \ldots, n$) are \begin{align}\label{eq:popln&comp_models}
    y_i &= X_i \beta_0 + \epsilon_{0i}, \quad y_i = X_i \beta + \epsilon'_{i}, \quad 
    \epsilon_{0i} \sim \mathcal{N}(0, \tau_0^2 V_{0i}), \quad  \epsilon'_{i} \sim \mathcal{N}(0, \tau_0^2 C_{i}), 
\end{align}
where $V_{0i} = Z_i \Sigma_0 Z_i^{\T} + I_{m_i}$, $C_{i} = Z_i S^{\T} \Gamma R R^{\T}\Gamma^{\T}S Z_i^{\T} + I_{m_i}$, and $\beta_0, \Sigma_0, \tau_0^2$ are the true parameter values. We define $C =
\diag(C_1, \ldots, C_n)$, $y \in \RR^N$ with $i$th row block $y_i$, and $X \in
\RR^{N \times p}$ with $i$th row block $X_i$. The compressed posterior
density of $\beta$ given $y, X, \delta^2$, and $\Lambda$ is
\begin{align}
  \label{eq:cond-beta}
  \beta \mid C, y, \Lambda, \delta^2 \sim \Ncal \left( \{ X^{\T} C^{-1} X + (\delta^{2} \Lambda)^{-1}\}^{-1} X^{\T} C^{-1} y, \tau_0^{2} \{ X^{\T} C^{-1} X + (\delta^{2} \Lambda)^{-1}\}^{-1}  \right).
\end{align}
If $\bar \beta (\phi)$ denotes the posterior mean of $\beta$ with $\phi = \{\Gamma, \Lambda, \delta\}$, then the prediction risk of the compressed posterior distribution is 
\begin{equation}\label{eq:risk_defn}
\frac{1}{N}\EE \| X \beta_0 - X \bar \beta
(\phi) \|_2^2=\frac{1}{N}\EE_{R,S,Z, X} \EE_y \left\{ \EE_{\phi \mid y}
\left( \| X \beta_0 - X \bar \beta (\phi) \|_2^2 \right) \right\},
\end{equation}
where $\EE_{R,S,Z,X}, \EE_{y}, \EE_{\phi \mid y}$ respectively denote the expectations with respect to
the distributions of $(R, S, Z, X), y$ and the conditional distribution of $\phi$ given $y$.

Establishing high-dimensional asymptotic properties for posteriors in Bayesian mixed-effects models is challenging for three main reasons: (a) the structural and algebraic complexity due to the presence of the additional random effect terms and the associated covariance parameters, (b) the high-dimensionality, and (c) the additional complexity of handling the posterior distribution, which as compared to a frequentist setting, adds another layer of expectation ($\EE_{\phi \mid y}$) to be simplified in the current analysis. While \cite{Lietal21} establish high-dimensional consistency results for mixed-effects models in the frequentist paradigm, such results have not yet been established in a Bayesian high-dimensional framework to the best of our knowledge. We are able to establish such results here with significant effort, but the regularity assumptions needed are slightly stronger than those needed for analogous results in the regression and the frequentist mixed model settings. This is, of course, not surprising given the discussion above.  

We now describe the regularity assumptions needed for our main asymptotic result (Theorem \ref{mainTheorem}). First, for simplicity, we assume that the error variance is known and the local shrinkage parameters are supported on a compact domain. 
\begin{enumerate}[label=(A\arabic*)]
  \item The error variance in \eqref{eq:compmem-varsel-eq1} equals $\tau_0^2$ and   $\lambda_j^2$ ($j=1, \ldots, p$) in \eqref{gbt-prior} are supported on $[a, 1/a]$, where $a$ is a universal constant. 
\end{enumerate}
The global shrinkage $\delta^2$ parameter in \eqref{gbt-prior} can depend on $N, p, q$ \citep{vanderpas17,Vanetal17}. The next assumption absorbs the CME model's dependence on $k_1, k_2$ into one ``compression dimension'' $k$ such that $k_1, k_2, k$ have the same asymptotic order. 
\begin{enumerate}[label=(A\arabic*), resume]
  \item There exists a $k$ and universal positive constants $\underline c_{1}, \underline c_{2}, \overline c_{1}, \overline c_{2}$ such that $\underline c_{j} k \leq k_{j} \leq \overline c_{j} k$ for $j=1, 2$. 
\end{enumerate}
The parameter $\Gamma$ in \eqref{eq:compmem-varsel-eq1} models the ``sparsity'' of the random effects covariance matrix $\Sigma_0$ in \eqref{eq:popln&comp_models}. We model $\Sigma_0$ as a low rank matrix by assuming that $\Gamma$ belongs to a class of matrices with a bounded operator norm. 
\begin{enumerate}[label=(A\arabic*), resume]
\item The support for the prior on $\text{vec}(\Gamma)$ in \eqref{eq:gen-lik-2} is restricted to the set $\{ \Gamma \in \RR^{k \times k}: \|\Gamma\| \leq b \}$ for a universal constant $b$, where $\| \cdot\|$ is the operator norm.
\end{enumerate}
The next assumption is about the distribution of the fixed and random effects covariates. 
\begin{enumerate}[label=(A\arabic*), resume]
\item For $i = 1, \ldots, n$, $X_i, Z_i, \epsilon_i$ in \eqref{eq:compmem-varsel-eq1} are mutually independent, entries of $X_i$ are independent and identically distributed as $\Ncal(0, \sigma_X^2)$ for a constant $\sigma_X^2$, and entries of $Z_i$ are independent and identically distributed as $\Ncal(0, \sigma_Z^2)$, where $\sigma_Z$
satisfies $\sigma_Z^4 k^{-4} n (q^{6} +m_{\max}^2 )= O(1)$ and $m_{\max} = \max\{m_1, \dots, m_n\}$.
\end{enumerate}

It is worth noting that the theoretical results in \cite{Lietal21}, for a frequentist setting, are derived {\it conditional on} $Z_1, Z_2, \dots, Z_n$. On the other hand, we take a different approach, where quantities of interest, such as the prediction risk, are averaged over the distribution of $Z_1, Z_2, \dots, Z_n$. Our theoretical results remain unchanged if the entries of $Z_i$ are independent sub-Gaussians with variance parameter $\sigma_Z^2$. The condition on  $\sigma_Z$ in  Assumption (A4) says that the variance of $Z_i$ entries rapidly decays to 0 as $n,q,m_{\max} = \max\{m_1, \dots, m_n\}$ tend to infinity. While this assumption is strong, it ensures that the moments of the operator norm of the random effects covariance matrix are bounded. Finally, the Gaussianity of the entries of $X_i$ cannot be relaxed to 
sub-Gaussianity, since the proof of Theorem \ref{mainTheorem} requires uniform boundedness of the moments of the condition number of $X$. While such bounds can be obtained for Gaussian entries in the current setting, they are not in general available for sub-Gaussian entries \citep{KSS:2025}.

The following theorem uses these assumptions to show that the prediction risk of the compressed posterior decays to zero.
\begin{theorem}
\label{mainTheorem}
If Assumptions (A1)-(A4) hold, $p = o(N)$, $k^2 \log k = o(p)$, $k^2 \log \log N = o(p)$, and $\|\beta_0 \|^2 = o(N)$,  then the posterior predictive risk satisfies
\begin{align*}
    \frac{1}{N}\EE \| X \beta_0 - X \bar \beta (\phi) \|_2^2 &\leq  \EE_{X}\left\{ \kappa^4(X)\right\} \Bigg[ \frac{2\| \beta_{0}\|_2^{2}}{N}\left\{O\left(\frac{nb^8\sigma_Z^8 q^{12}}{k^8} + \frac{nb^8\sigma_Z^8 m_{\max}^4}{k^8} + nb^8\sigma_Z^8 \right)\right\} + \nonumber \\
    & \qquad \qquad \qquad  \qquad  \frac{4\tau_{0}^{2}}{a^{4}} \Bigg\{ \left(\frac{N-p}{N\log N}\right)O\left( \frac{nb^5\sigma_Z^6 (q^9 + m_{\max}^3)}{k^6}\right) + \nonumber\\
        &\qquad  \qquad  \qquad  \qquad  \left(\frac{p + 4e^{-p/8}}{N}\right)O\left(\frac{nb^3\sigma_Z^4(q^6+m_{\max}^2)}{k^4}\right)\Bigg\}\Bigg] = o(1),
\end{align*}
where $\kappa(X)$ is the condition number of $X$, and $m_{\max} = \max\{m_1, \dots, m_n\}$. 
\end{theorem}
The proof of this theorem is in the supplementary material, along with other proofs. Theorem~\ref{mainTheorem} implies that the compression dimension $k$ is chosen to be very small compared to $p$ and $N$ in practice. This enables sampling of the compressed covariance parameter $\Gamma$ in $O(k^6)$ complexity, resulting in significant gains in computational efficiency of the Gibbs sampler in Algorithm \ref{algo:FES}. 

We do not provide any results for the fixed effects selection because the literature on variable selection properties of the Horseshoe prior is still evolving \citep{Vanetal17,SonLia23}; however, we expect these optimality properties to hold in our setting because our assumptions ensure that the spectrum of the random effects covariance matrices is uniformly bounded away from zero and infinity. The excellent performance of the CME model in fixed effects selection in our simulations further supports this conclusion.

\section{Simulation Studies}\label{Section:Simulations}

\subsection{Experimental Setup}\label{Section:Sim_setup}

Our simulation setup is based on \cite{Lietal21}, which is a penalized quasi-likelihood competitor of the CME model. We set $p = 300, q = 14, n = 36$, and $m_1 = \dots = m_n = m$ with $m \in \{4, 8, 12\}$. The entries of $X_i$'s and $Z_i$'s are independently and identically distributed as $\mathcal{N}(0, 1)$. The first five entries of $\beta$ are set to $(1, 0.5, 0.2, 0.1, 0.05)^{\T}$ and the remaining entries are zero. To understand the impact of compressing the covariance matrix, we consider three different choices of $\Sigma$. First, $\Sigma$ is diagonal and positive semi-definite, with  $\Sigma_{ii} = 0.5$ for $1 \leq i \leq \lceil q/2 \rceil$, and $\Sigma_{ij} = 0$ for $i \neq j$. 
Second, $\Sigma$ is block-diagonal with $\Sigma = LL^{\T}$, where $L$ has rank $\lceil \log q \rceil$. The first, second, and third columns of $L$ contain nonzero entries independently generated from continuous uniform distribution [0, 3] in rows 1-5, 6-10, and 11-14, respectively. Finally, $\Sigma$ is a Toeplitz matrix, with $\Sigma_{i j} = 0.5 ^ {|i - j|}$ for any $i, j \in \{1, \dots, q\}$. The first two $\Sigma$ choices impose different low-rank assumptions, whereas the third $\Sigma$ is positive definite. The three $\Sigma$ and three $m$ choices lead to nine different simulation setups. We perform 50 replications for every simulation setup.

We used the Algorithm \ref{algo:FES} for fixed effects selection and prediction in the CME model. To assess the sensitivity of our results to the choice of compression dimensions $k_1$ and $k_2$, we evaluated the empirical performance for $k_1, k_2 \in \{\lceil \log q \rceil, \lceil q/2 \rceil, q \}= \{3, 7, 14 \}$. The prior specification for $(\beta, \tau^2)$ followed from \eqref{gbt-prior}, with  $a_0 = 0.01, b_0 = 0.01$, and $\sigma^2_\gamma = 1$.  We ran the Gibbs sampler in Algorithm \ref{algo:FES} for 15,000 iterations and discarded the first 5,000 samples as burn-ins. We used the $S_2M$ algorithm to estimate the zero and nonzero $\beta$ components from the posterior $\beta$ draws \citep{LiPati17}. 

We compared the CME model's performance in fixed effects selection and prediction against that of penalized quasi-likelihood (PQL) methods and a Bayesian oracle. The PQL competitors were \citet{FanLi12} (PQL-1) and \citet{Lietal21} (PQL-2). Following \citet{Lietal21}, we set the tuning parameter $\lambda$ in PQL-2 to $\hat{\tau} \{2 \log(p)/N\}^{1/2}$, where $\hat{\tau}$ is the estimated error variance based on scaled lasso \citep{SunZha12}. While PQL-1 had no debiasing step, we debiased its $\beta$ estimates for $\lambda = \lambda_{\text{min}}$ to construct the confidence intervals, where $\lambda_{\text{min}}$ achieves the minimum mean cross-validation error. The Bayesian oracle (OracleHS) fixed $\Sigma$ at its true value and used the same prior on $(\beta, \tau^2)$ as specified in the CME model. The OracleHS method selected variables after scaling the error variances to the homoscedastic setting \citep{2010_PolsonScott}. Following CME, we ran the OracleHS sampler for 15,000 iterations, discarded the first 5,000 samples as burn-ins, and used the  $S_2M$ algorithm for fixed effects selection. 

We used several metrics to compare empirical performance. Fixed-effects selection was evaluated using the coverage probability and mean width of 95\% credible or confidence intervals, true positive rate, and false positive rate. For predictive performance, we generated a new test dataset consisting of 12 subjects, each with $m$ observations. The mean squared prediction error (MSPE) was computed as $\sum_{i=1}^{n}\{\sum_{j=1}^{m_i}(y^{*}_{ij} - \hat{y}_{ij})^2 / m_i\} / n$, where $y^{*}_{ij}$ and $\hat{y}_{ij}$ respectively denote the true and predicted responses. Additionally, we assessed posterior predictive performance by comparing the coverage probability and mean width of the 95\% posterior predictive intervals obtained from the CME and OracleHS models.

\subsection{Empirical Results}\label{Section:Sim_results}

The CME model with small values of compression dimensions, $k_1$ and $k_2$, has competitive fixed effects selection performance across all simulation settings. The CME model with $k_1 = k_2 = 3$ yields the narrowest credible intervals for the nonzero components of $\beta$, while maintaining empirical coverage close to the nominal 95\% level.  (Figures \ref{fig:VS_covg_allSigma_IndepX} and \ref{fig:VS_width_allSigma_IndepX}). For nonzero $\beta$ coefficients with small magnitude (e.g., 0.05), the empirical coverage of the CME model approaches the nominal level as the sample size increases. When $\Sigma$ is low-rank, the coverage of PQL-1 model for the nonzero components of $\beta$ is below the nominal level. The CME and PQL-1 models have similar coverage when $\Sigma$ is full-rank, likely because PQL-1 replaces the random effects covariance matrix with a full-rank diagonal proxy matrix. On the other hand, PQL-2 model produces the widest confidence intervals for all $\Sigma$ choices. The credible intervals produced by the CME and OracleHS methods have similar properties, with both achieving 99\% empirical coverage for the zero components of $\beta$. For small $k_1$ and $k_2$ (e.g., $k_1 = k_2 = 3$), the CME model outperforms its competitors in terms of true positive and false positive rates as the sample size increases with $m$ (Table \ref{Table_tprfpr_allSigma_IndepX}).

The CME model outperforms its penalized quasi-likelihood competitors in predictive performance across all simulation settings (Figure \ref{fig:rel_mspe_allSigma_IndepX}).  While the CME and PQL-1 models have comparable mean square prediction errors when the true $\Sigma$ is a Toeplitz matrix and the sample size is the smallest, the CME model's performance improves with increasing $m$. Furthermore, the CME model's posterior predictive intervals for $k_1 = k_2 = 3$ are the shortest among all compression dimensions and achieve the nominal coverage level. Their lengths are also nearly identical to those of the equivalent OracleHS model  (Table \ref{Table_covg&relWidth_PI_CME_allSigma_IndepX}). In contrast, the penalized quasi-likelihood methods fail to estimate the variance components if $q > m$.

In summary, as the sample size increases with $m$, the CME model with small compressed dimensions (e.g., $k_1 = k_2 = 3$) outperforms its quasi-likelihood competitors across all simulation settings in fixed effects selection and prediction. Additionally, in the supplementary materials, we consider a correlated fixed effects setting where each row of $X \in \mathbb{R}^{N \times p}$ follows a multivariate normal distribution with mean zero and a Toeplitz covariance structure. The results in this setting agree with those reported here.

\begin{figure}[h!]
    \centering
    \includegraphics[scale = 0.50]{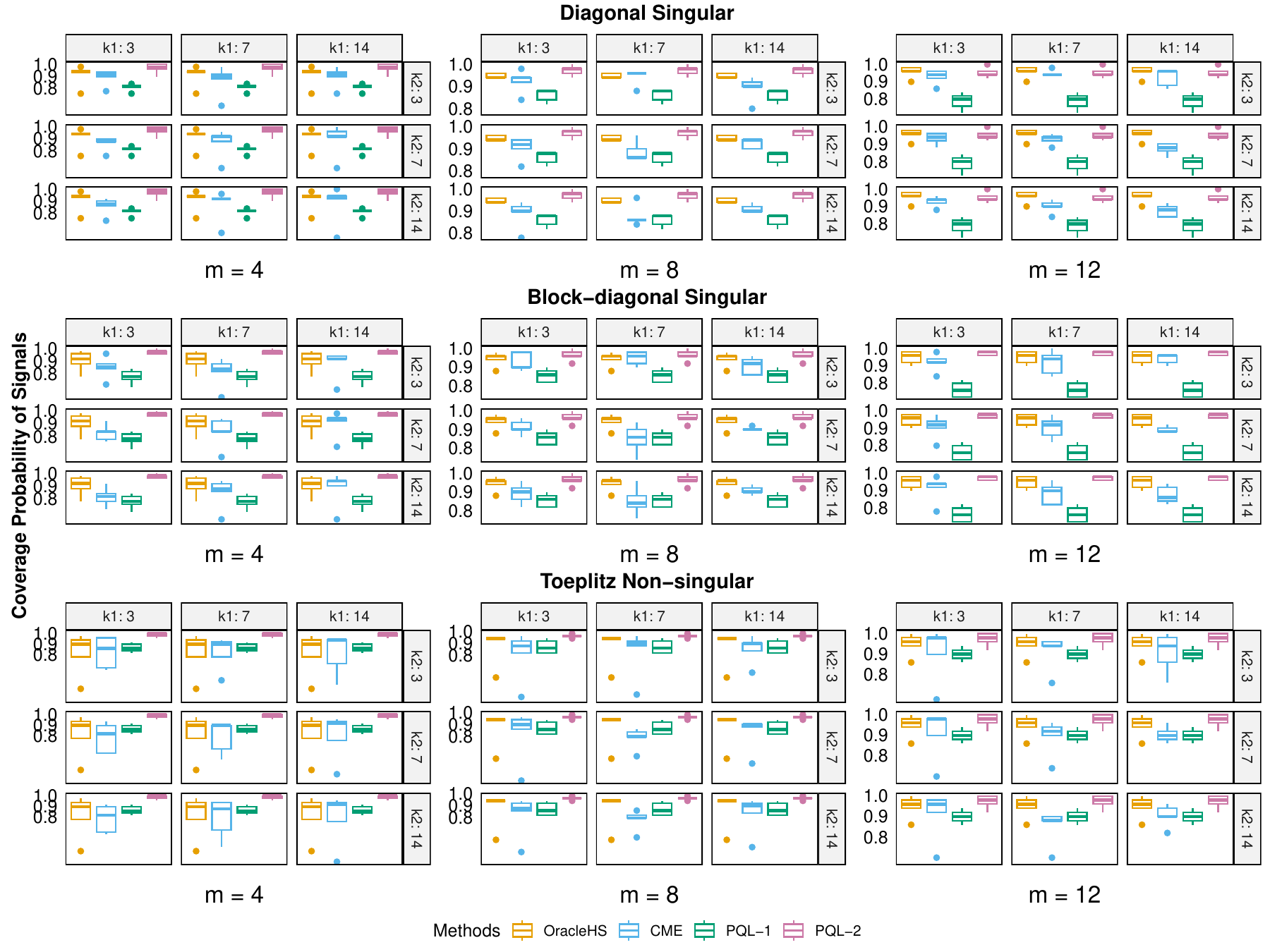}
    \caption{Coverage of signals with varying cluster size $m \in \{4, 8, 12\}$ and compression dimensions ($k_1 \in \{3, 7, 14\}$ and $k_2 \in \{3, 7, 14 \}$).   
  CME outperforms its penalized quasi-likelihood competitors with random effects covariance matrices as diagonal (singular), block-diagonal (singular), and Toeplitz (non-singular). OracleHS, oracle with Horseshoe prior; CME, compressed mixed-effects model; PQL-1, the penalized quasi-likelihood approach of \citet{FanLi12}; PQL-2, the penalized quasi-likelihood approach of \citet{Lietal21}.}
    \label{fig:VS_covg_allSigma_IndepX}
\end{figure}

\begin{figure}[h!]
    \centering
    \includegraphics[scale = 0.50]{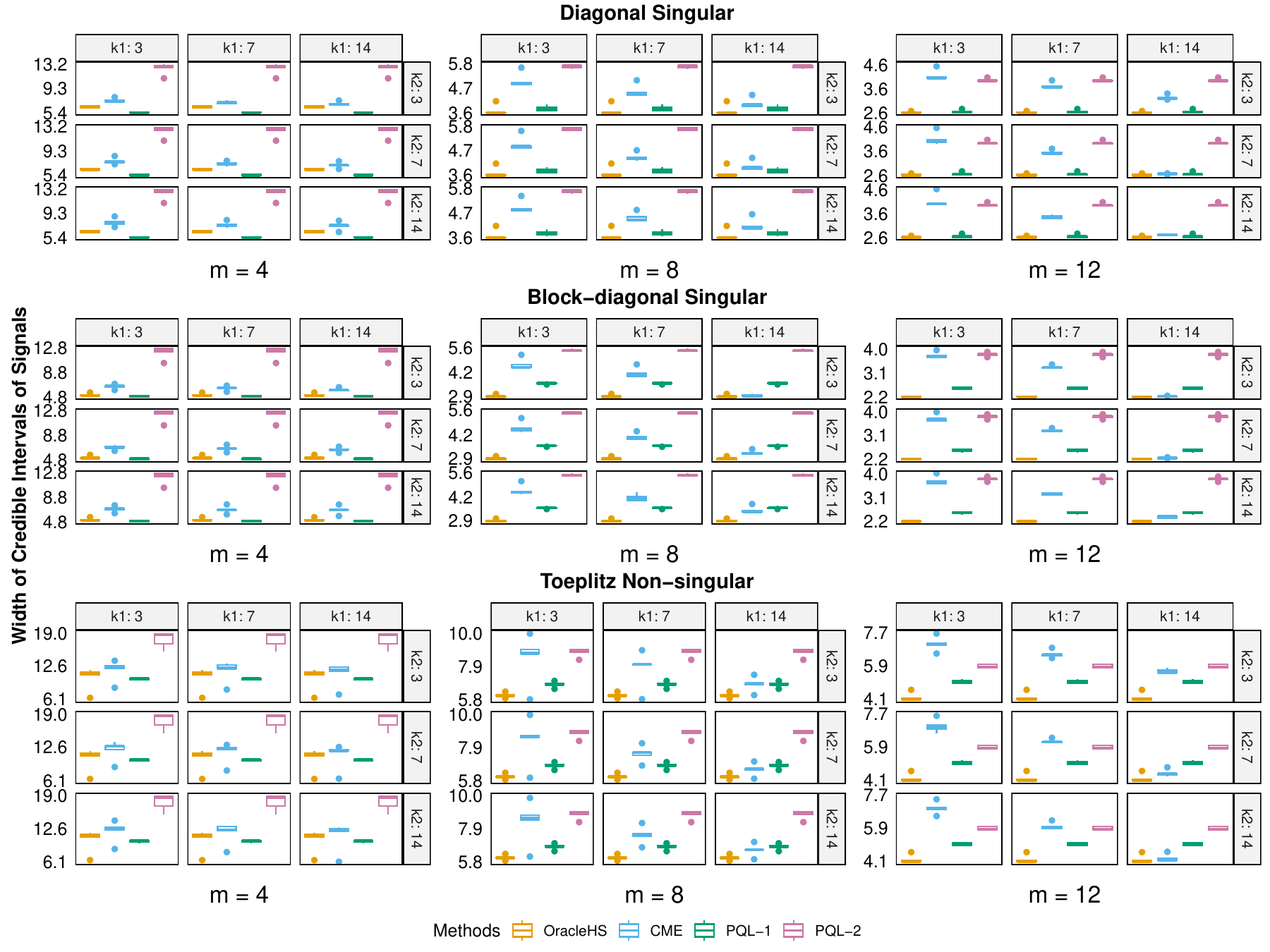}
    \caption{Width of credible intervals of signals depending on cluster size $m \in \{4, 8, 12\}$ and compression dimensions ($k_1 \in \{3, 7, 14\}$ and $k_2 \in \{3, 7, 14 \}$).   
  CME outperforms its penalized quasi-likelihood competitors with random effects covariance matrices as diagonal (singular), block-diagonal (singular), and Toeplitz (non-singular). OracleHS, oracle with Horseshoe prior; CME, compressed mixed-effects model; PQL-1, the penalized quasi-likelihood approach of \citet{FanLi12}; PQL-2, the penalized quasi-likelihood approach of \citet{Lietal21}.}
    \label{fig:VS_width_allSigma_IndepX}
\end{figure}

\begin{table}[h!]
\centering
\caption{True positive rates (TPR) and false positive rates (FPR in parentheses) averaged across 50 replicates in fixed effects selection, with varying cluster size $m \in \{4, 8, 12\}$, compression dimensions ($k_1 \in \{3, 7, 14\}$ and $k_2 \in \{3, 7, 14 \}$), and random effects covariance matrices as diagonal positive semi-definite, block-diagonal positive semi-definite, and Toeplitz positive definite. OracleHS, oracle with Horseshoe prior; CME, compressed mixed-effects model; PQL-1, the penalized quasi-likelihood approach of \citet{FanLi12}; PQL-2, the penalized quasi-likelihood approach of \citet{Lietal21}.}
\resizebox{\columnwidth}{!}{
\begin{tabular}{|c|c|cccc|ccc|ccc|}
\hline
\multirow{3}{*}{$\Sigma$} & \multirow{3}{*}{Methods} & \multirow{3}{*}{$k_1$} & \multicolumn{9}{c|} {$k_2$}\\ \cline{3-12}
  & & & \multicolumn{3}{c|} {3} & \multicolumn{3}{c|} {7} & \multicolumn{3}{c|} {14}\\ \cline{3-12}
& & & $m = 4$ & $m = 8$ & $m = 12$ & $m = 4$ & $m = 8$ & $m = 12$ & $m = 4$ & $m = 8$ & $m = 12$\\
\hline
\multirow{8}{*}{Diagonal psd} & & & & & & & & & & & \\
& \multirow{3}{*}{CME} & 3 & \textbf{0.95 (0.09)} & \textbf{0.98 (0.04)} & \textbf{0.99 (0.02)} & 0.96 (0.17) & 0.99 (0.04) & 0.99 (0.02) & 0.94 (0.16) & 0.98 (0.08) & 1 (0.03) \\
&  & 7 & 0.95 (0.10) & 0.99 (0.04) & 1 (0.02) & 0.94 (0.06) & 1 (0.18) & 1 (0.04) & 0.92 (0.04) & 1 (0.17) & 0.99 (0.22) \\
&  & 14 & 0.96 (0.04) & 0.99 (0.17) & 1 (0.02) & 0.93 (0.02) & 1 (0.04) & 1 (0.11) & 0.8 (0.03) & 0.96 (0.01) & 1 (0.06) \\
&  OracleHS & & 0.97 (0.03) & 0.99 (0.01) & 1 (0) & 0.97 (0.03) & 0.99 (0.01) & 1 (0) & 0.97 (0.03) & 0.99 (0.01) & 1 (0) \\
& PQL-1 & & 0.97 (0.09) & 1 (0.14) & 1 (0.32) & 0.97 (0.09) & 1 (0.14) & 1 (0.32) & 0.97 (0.09) & 1 (0.14) & 1 (0.32)\\
& PQL-2 & & 0.66 (0) & 0.91 (0) & 1 (0) & 0.66 (0) & 0.91 (0) & 1 (0) & 0.66 (0) & 0.91 (0) & 1 (0) \\
& & & & & & & & & & &\\
\hline
\multirow{8}{*}{Block-diagonal psd} & & & & & & & & & & & \\
& \multirow{3}{*}{CME} & 3 & \textbf{0.96 (0.09)} & \textbf{0.98 (0.03)} & \textbf{0.99 (0.02)} & 0.96 (0.16) & 0.99 (0.04) & 1 (0.02) & 0.94 (0.16) & 0.99 (0.08) & 1 (0.03) \\ 
& & 7 & 0.94 (0.08) & 1 (0.03) & 1 (0.01) & 0.94 (0.06) & 0.99 (0.18) & 1 (0.03) & 0.91 (0.04) & 0.99 (0.16) & 1 (0.18) \\
 & & 14 & 0.96 (0.03) & 1 (0.1) & 1 (0) & 0.91 (0.02) & 1 (0.03) & 1 (0.11) & 0.8 (0.03) & 0.98 (0.01) & 1 (0.07) \\
 & OracleHS & & 0.97 (0.05) & 1 (0.05) & 1 (0.05) & 0.97 (0.05) & 1 (0.05) & 1 (0.05) & 0.97 (0.05) & 1 (0.05) & 1 (0.05) \\
 & PQL-1 & & 0.97 (0.1) & 1 (0.14) & 1 (0.36) & 0.97 (0.1) & 1 (0.14) & 1 (0.36) & 0.97 (0.1) & 1 (0.14) & 1 (0.36)\\
 & PQL-2 & & 0.68 (0) & 0.96 (0) & 1(0) & 0.68 (0) & 0.96 (0) & 1(0) & 0.68 (0) & 0.96 (0) & 1(0) \\
  & & & & & & & & & & &\\
 \hline
 \multirow{8}{*}{Toeplitz pd} & & & & & & & & & & & \\
& \multirow{3}{*}{CME} & 3 & \textbf{0.89 (0.23)} & \textbf{0.90 (0.06)} & \textbf{0.96 (0.05)} & 0.90 (0.31) & 0.91 (0.08) & 0.95 (0.05) & 0.88 (0.3) & 0.92 (0.15) & 0.95 (0.07) \\ 
& & 7 & 0.89 (0.19) & 0.91 (0.08) & 0.96 (0.05) & 0.88 (0.11) & 0.95 (0.29) & 0.97 (0.08) & 0.86 (0.07) & 0.95 (0.26) & 0.96 (0.31)\\
 & & 14 & 0.86 (0.07) & 0.95 (0.29) & 0.98 (0.06) & 0.87 (0.04) & 0.94 (0.06) & 0.99 (0.2) & 0.72 (0.03) & 0.94 (0.03) & 0.99 (0.06) \\
 & OracleHS & & 0.91 (0.06) & 0.96 (0.04) & 0.98 (0.01) & 0.91 (0.06) & 0.96 (0.04) & 0.98 (0.01) & 0.91 (0.06) & 0.96 (0.04) & 0.98 (0.01) \\
& PQL-1 & & 0.88 (0.07) & 0.97 (0.10) & 1 (0.18) & 0.88 (0.07) & 0.97 (0.10) & 1 (0.18) & 0.88 (0.07) & 0.97 (0.10) & 1 (0.18)\\
 & PQL-2 & & 0.58 (0) & 0.80 (0) & 0.94 (0) & 0.58 (0) & 0.80 (0) & 0.94 (0) & 0.58 (0) & 0.80 (0) & 0.94 (0) \\
 & & & & & & & & & & &\\
\hline
 \end{tabular}
}
\label{Table_tprfpr_allSigma_IndepX}
\end{table}

\begin{table}[h!]
\centering
\caption{Coverage of 95\% Prediction Intervals produced by the CME model with their relative width in parentheses. Relative width is computed by dividing the width of a prediction interval produced by CME model by that of the OracleHS, the oracle with Horseshoe prior. All metrics are averaged over 50 replications.}
\resizebox{\columnwidth}{!}{
\begin{tabular}{|c|cccc|ccc|ccc|}
\hline
\multirow{3}{*}{$\Sigma$} & \multirow{3}{*}{$k_1$} & \multicolumn{9}{c|} {$k_2$}\\ \cline{2-11}
  & & \multicolumn{3}{c|} {3} & \multicolumn{3}{c|} {7} & \multicolumn{3}{c|} {14}\\ \cline{2-11}
& & $m = 4$ & $m = 8$ & $m = 12$ & $m = 4$ & $m = 8$ & $m = 12$ & $m = 4$ & $m = 8$ & $m = 12$\\
\hline
\multirow{5}{*}{Diagonal psd} & & & & & & & & & & \\
 & 3 & \textbf{0.95 (1.16)} & \textbf{0.95 (1.09)} & \textbf{0.95 (1.07)} & 0.96 (1.35) & 0.96 (1.16) & 0.96 (1.11) & 0.96 (1.42) & 0.97 (1.32) & 0.97 (1.22) \\
& 7 & 0.96 (1.32) & 0.96 (1.20) & 0.96 (1.15) & 0.97 (1.37) & 0.96 (1.40) & 0.98 (1.28) & 0.97 (1.39) & 0.96 (1.42) & 0.98 (1.47) \\
& 14 & 0.98 (1.31) & 0.97 (1.35) & 0.98 (1.30) & 0.98 (1.31) & 0.98 (1.29) & 0.98 (1.33) & 0.99 (1.36) & 0.98 (1.31) & 0.98 (1.33) \\
& & & & & & & & & &\\
\hline
\multirow{5}{*}{Block-diagonal psd} & & & & & & & & & & \\
 & 3 & \textbf{0.94 (1.22)} & \textbf{0.95 (1.12)} & \textbf{0.95 (1.10)} & 0.95 (1.40) & 0.95 (1.19) & 0.95 (1.15) & 0.96 (1.47) & 0.96 (1.35) & 0.96 (1.26) \\ 
& 7 & 0.96 (1.35) & 0.96 (1.22) & 0.96 (1.17) & 0.97 (1.41) & 0.97 (1.46) & 0.97 (1.30) & 0.98 (1.44) & 0.97 (1.49) & 0.98 (1.50) \\
 & 14 & 0.97 (1.32) & 0.97 (1.29) & 0.98 (1.28) & 0.98 (1.35) & 0.97 (1.25) & 0.98 (1.31) & 0.98 (1.41) & 0.97 (1.29) & 0.98 (1.29) \\
  & & & & & & & & & &\\
 \hline
 \multirow{5}{*}{Toeplitz pd} & & & & & & & & & & \\
 & 3 & \textbf{0.95 (1.20)} & \textbf{0.95 (1.10)} & \textbf{0.96 (1.07)} & 0.95 (1.28) & 0.96 (1.17) & 0.96 (1.12) & 0.96 (1.29) & 0.97 (1.32) & 0.97 (1.23) \\ 
& 7 & 0.96 (1.24) & 0.96 (1.21) & 0.97 (1.15) & 0.97 (1.25) & 0.96 (1.34) & 0.98 (1.29) & 0.96 (1.24) & 0.96 (1.34) & 0.99 (1.46)\\
 & 14 & 0.97 (1.21) & 0.96 (1.32) & 0.98 (1.30) & 0.98 (1.19) & 0.97 (1.21) & 0.98 (1.27) & 0.98 (1.20) & 0.97 (1.21) & 0.98 (1.22) \\
 & & & & & & & & & &\\
\hline
 \end{tabular}
}
\label{Table_covg&relWidth_PI_CME_allSigma_IndepX}
\end{table}

\begin{figure}[h!]
    \centering
    \includegraphics[scale = 0.50]{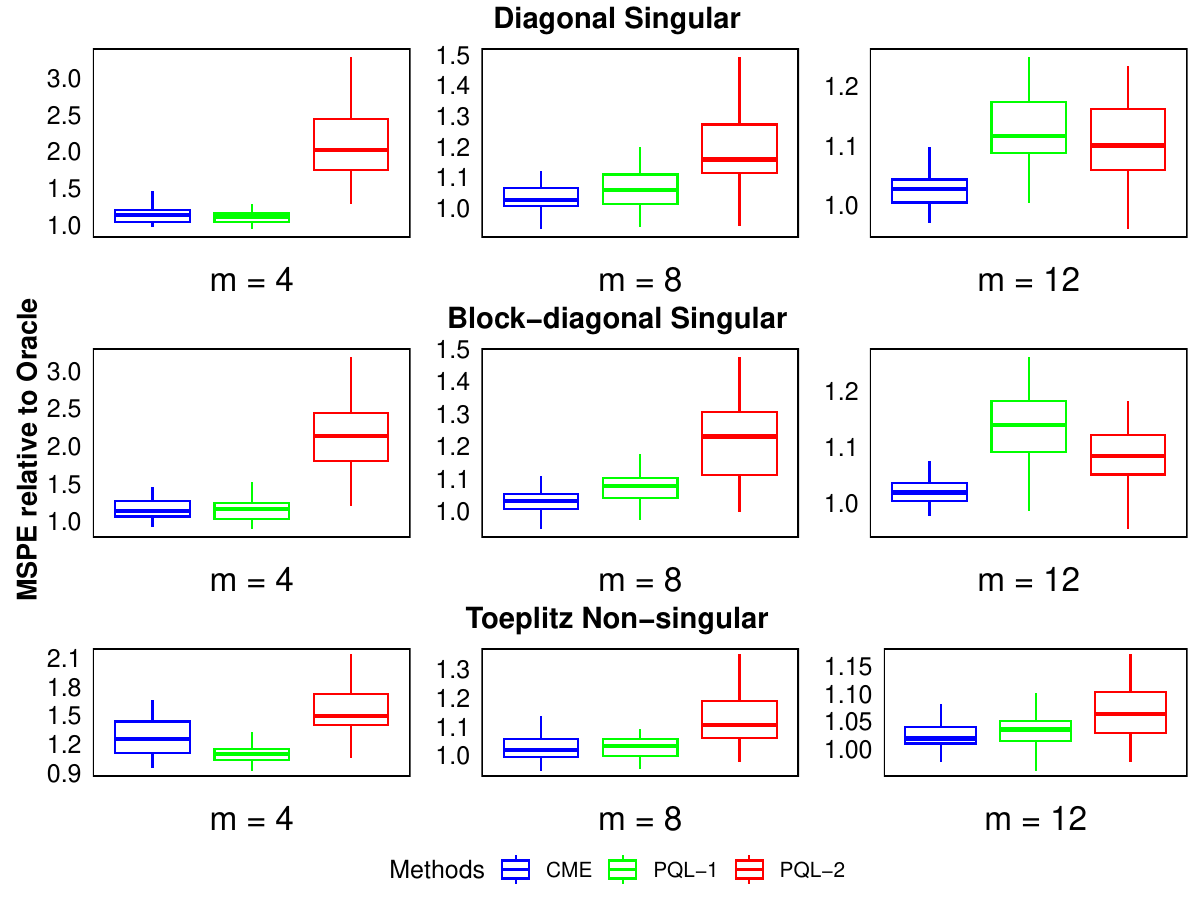}
    \caption{Relative mean square prediction error (MSPE) averaged over 50 replications. Relative MSPE is computed by dividing the MSPE for a method by that of OracleHS, the oracle with Horseshoe prior. CME, compressed mixed-effects model; PQL-1, the penalized quasi-likelihood approach of \citet{FanLi12}; PQL-2, the penalized quasi-likelihood approach of \citet{Lietal21}. Results for the CME model are presented for the best choice of compression dimensions $k_1 = k_2 = 3$.}
    \label{fig:rel_mspe_allSigma_IndepX}
\end{figure}

\section{Real Data Analysis}\label{Section:RealData}

We use the riboflavinV100 data from \citet{2014_buhlman_riboflavin} to evaluate the performance of the CME model on a real-world data. This data has been used for benchmarking the performance of high-dimensional mixed models \citep{2020_Bradetal, 2023_Oliveira_etal}. It contains information about $n = 28$ strains of Bacillus subtilis bacterium, found in the human gastrointestinal tract, that produces Riboflavin (vitamin B2). Each bacterium strain has two to four observations at different time points; that is, $m_i \in \{2, 3, 4\}$ and $N = \sum_{i=1}^{28} m_i = 71$. For a strain $i$, $y_i \in \RR^{m_i}$ is the logarithm of riboflavin production rate and $\tilde X_i \in \RR^{m_i \times 100}$ is the log-transformed expression levels of the 100 genes affecting $y_i$. We standardize the covariates to have mean 0 and variance 1.

We defined the covariates in \eqref{eq:mem1} by taking $X_i \in \RR^{m_i \times 104}$ and $Z_i = X_i$, where the first $X_i$ column contains $1$s, the next hundred $X_i$ columns equal the standardized $\tilde X_i$, and the last three $X_i$ columns contain three B-spline basis functions evaluated at the $m_i$ time points. We used the \texttt{splines2} package in R for computing the last three columns of $X_i$ and $Z_i$ \citep{2024_splines2-package}. 

We applied the CME model with compressed dimensions $k_1, k_2 \in \{2, 4, \log(104) \approx 5 \}$. The setup for the penalized competitors PQL-1 and PQL-2 remained unchanged from Section \ref{Section:Sim_setup}. We did not include the OracleHS method in this analysis because the marginal covariance matrix was unknown. For model training, we randomly selected 21 strains and evaluated empirical performance on the remaining 7 strains. We used the metrics described in Section \ref{Section:Sim_setup} for comparisons, excluding those for fixed effects selection due to the lack of ground truth. This train-test split procedure was repeated 50 times to assess the variability of the results. The CME model with small compression dimensions outperforms both PQL-1 and PQL-2 in predictive performance. Among all compression dimensions, the CME model with $k_1 = k_2 = 2$ yields the narrowest 95\% posterior predictive intervals while achieving an average empirical coverage of 0.98 and the lowest mean squared prediction error (Figure \ref{fig:ribo_PostPred} and Table \ref{table_mspe_ribo}). This CME model also selects a more parsimonious subset of genes than the PQL-1 method and identifies genes like XHLB\_at and GAPB\_at that are often found to enhance riboflavin production in microbial and synthetic biology contexts (Table \ref{Table_genes_ribo}). These findings, combined with our simulation results, demonstrate that the CME model with compression dimensions $k_1=2$ and $k_2=2$ outperforms its quasi-likelihood competitors PQL-1 and PQL-2 in both prediction and fixed-effects selection on the riboflavinV100 dataset.
\begin{table}[h!]
    \centering
    \caption{Mean square prediction error averaged across 50 splits of riboflavinV100 data with Monte Carlo error in parentheses. CME, compressed mixed-effects model; PQL-1, the penalized quasi-likelihood approach of \citet{FanLi12}; PQL-2, the penalized quasi-likelihood approach of \citet{Lietal21}. Results for the CME model are presented for the best choices of compression dimensions $k_1 = k_2 = 2$.}
    \begin{tabular}{|c|c|}
    \hline
    Methods & MSPE\\ \hline
        CME & \textbf{0.83} (0.46) \\
        PQL-1 & 0.96 (0.44)\\
        PQL-2 & 1.36 (0.45)\\ \hline
    \end{tabular}
    \label{table_mspe_ribo}
\end{table}

\begin{figure}[h!]
    \centering
    \includegraphics[scale = 0.65]{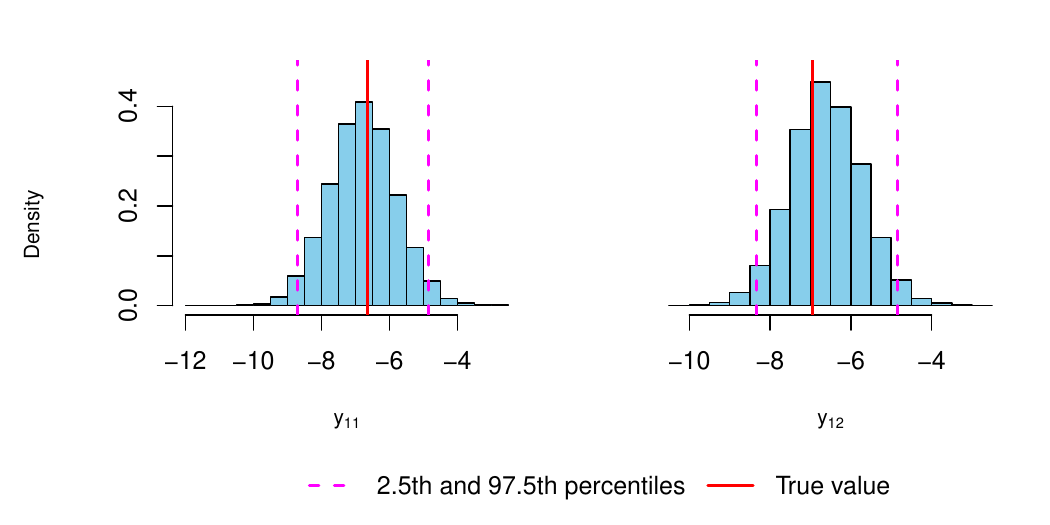}
    \caption{Posterior predictive uncertainty captured by our method with compressed dimensions $k_1 = k_2 = 2$.}
    \label{fig:ribo_PostPred}
\end{figure}

\begin{table}[h!]
    \centering
    \caption{Gene selections in riboflavinV100 data. CME, compressed mixed-effects model; PQL-1, the penalized quasi-likelihood approach of \citet{FanLi12}; PQL-2, the penalized quasi-likelihood approach of \citet{Lietal21}.}
    \begin{tabular}{|cc|}
     \hline
        Methods & Fixed effects list \\ \hline
        \multirow{2}{*}{CME} & YCKE\_at GUAB\_at YCGN\_at ABRB\_at \\ & XHLB\_at YTGA\_at GAPB\_at YXLD\_at\\
         & \\
        \multirow{2}{*}{PQL-1} & YCKE\_at GUAB\_at YCGN\_at ABRB\_at YCDH\_at YRBA\_at NADB\_at YPUF\_at \\ & DEGQ\_r\_at YXLD\_at YXLF\_at YXLE\_at ARGF\_at MTLA\_at YPSB\_at\\
         & \\
        \multirow{2}{*}{PQL-2} & YCGN\_at YTGA\_at GAPB\_at YXLD\_at \\ & XHLA\_at YCGM\_at YTIA\_at\\ 
    \hline
    \end{tabular}
    \label{Table_genes_ribo}
\end{table}

\section{Discussion}\label{Section:Discussion}

The CME model's strategy of compressing the covariance parameters generalizes beyond linear mixed models. First, the ultimate P\'olya-Gamma data-augmentation strategies enable compression in the generalized linear mixed models \citep{2013_Polson_etal,2024_Zens_etal}. Second, we are extending the covariance compression framework to accommodate separable covariance matrices, which arise naturally in array-variate generalizations of mixed models \citep{HulSri25}. Finally, we have used only one set of Gaussian random  matrices (i.e., $R$ and $S$) for compressing the row and column dimensions of the random effects covariance matrix for simplicity. The computational efficiency of the CME model improves if we use sparse random matrices \citep{MukDun20}. The prediction risk of the compressed model reduces further by averaging over multiple sets of random matrices.  This follows from a simple application of Jensen's inequality; see Section 4.3 in \citet{Sla18} for more details.

\section*{Declaration of the use of generative AI and AI-assisted technologies}
During the preparation of this work the authors used ChatGPT in order to check for grammatical errors. After using this tool the authors reviewed and edited the content as necessary and take full responsibility for the content of the publication.
\section*{Acknowledgement}
Sreya Sarkar and Sanvesh Srivastava were partially supported by grants from the National Institutes of Health (1DP2MH126377-01) and the National Science Foundation (DMS1854667). The authors thank Dale Zimmerman and Joyee Ghosh for their valuable feedback on an earlier version of the manuscript.   The code used in the experiments is publicly available at \href{https://github.com/SreyaSarkar21/BayesianCME}{https://github.com/SreyaSarkar21/BayesianCME}.

\newpage

\begin{center}
    \huge \textbf{Supplementary Material for Bayesian Compressed Mixed-Effects Models}
\end{center}
\appendix
\section{Theoretical Properties of Mean Square Prediction Risk}

\subsection{Proof of Theorem 3.1}\label{Section:ProofThm}

Recall the setup for establishing the theoretical properties of the compressed
mixed-effects model from Section \ref{Section:Theory} in the main manuscript. The true parameters are $\beta_0$, $\Sigma_0$, and $\tau_0^2$. The (true) data generating model and the (working) compressed model, respectively, are $y_i = X_i
\beta_0 + \epsilon_{0i}$ and $y_i = X_i \beta + \epsilon_i'$, where
$\epsilon_{0i} \sim \Ncal(0, \tau ^ 2_0 V_{0i})$, $V_{0i} = Z_i \Sigma_0 Z_i^\T
+ I_{m_i}$, $\epsilon_i' \sim \mathcal{N}(0, \tau ^ 2 C_i)$, $C_i = Z_i S^{\T}
\Gamma R R^{\T} \Gamma^{\T} S Z_i^{\T} + I_{m_i}$, and $i=1, \ldots, n$.

Define $C = \diag(C_1, \ldots, C_n)$, $y \in \RR^N$ with $i$th row block $y_i$, and $X \in \RR^{N \times p}$ with $i$th row block $X_i$. The compressed posterior
density of $\beta$ given $y, X, \delta^2$, and $\Lambda$ is
\begin{align}
  \label{eq:condbeta}
  \beta \mid C, y, \Lambda, \delta^2 \sim \Ncal \left( \{ X^{\T} C^{-1} X + (\delta^{2} \Lambda)^{-1}\}^{-1} X^{\T} C^{-1} y, \tau_0^{2} \{ X^{\T} C^{-1} X + (\delta^{2} \Lambda)^{-1}\}^{-1}  \right),
\end{align}
where the parameters $\delta^2$ and $\Lambda$ are defined in Sections \ref{Section:model} and \ref{Section:FES} of the main manuscript. 

The prediction risk of the compressed posterior is defined as
\begin{equation}\label{eq:predrisk_defn}
\frac{1}{N}\EE \| X \beta_0 - X \bar \beta
(\phi) \|_2^2=\frac{1}{N}\EE_{R,S, X, Z} \EE_y \left\{ \EE_{\phi \mid y}
\| X \beta_0 - X \bar \beta (\phi) \|_2^2  \right\},
\end{equation}
where the posterior mean of $\beta$ in \eqref{eq:condbeta} is denoted as $\bar \beta (\phi)$ with $\phi = \{\Gamma, \Lambda, \delta\}$, and $\EE_{R,S, X, Z}, \EE_{y}$, and $\EE_{\phi \mid y}$ respectively denote the expectations with respect to the distributions of $(R, S, X, Z), y$ and the conditional distribution of $\phi$ given $y$. 

We restate Theorem 3.1 from Section \ref{Section:Theory} from the main manuscript below, followed by a detailed proof. 

\begin{theorem}
\label{TheoremPredRisk}
If Assumptions (A1)-(A4) in Section \ref{Section:Theory} of the main manuscript hold, $p = o(N)$, $k^2 \log k = o(p)$, $k^2 \log \log N = o(p)$, and $\|\beta_0 \|^2 = o(N)$, then the posterior predictive risk satisfies
\begin{align*}
   \frac{1}{N}\EE \| X \beta_0 - X \bar \beta (\phi) \|_2^2 &\leq  \EE_{X}\left\{ \kappa^4(X)\right\} \Bigg[ \frac{2\| \beta_{0}\|_2^{2}}{N}\left\{O\left(\frac{nb^8\sigma_Z^8 q^{12}}{k^8} + \frac{nb^8\sigma_Z^8 m_{\max}^4}{k^8} + nb^8\sigma_Z^8 \right)\right\} + \nonumber \\
    & \qquad \qquad \qquad  \qquad  \frac{4\tau_{0}^{2}}{a^{4}} \Bigg\{ \left(\frac{N-p}{N\log N}\right)O\left( \frac{nb^5\sigma_Z^6 (q^9 + m_{\max}^3)}{k^6}\right) + \nonumber\\
        &\qquad  \qquad  \qquad  \qquad  \left(\frac{p + 4e^{-p/8}}{N}\right)O\left(\frac{nb^3\sigma_Z^4(q^6+m_{\max}^2)}{k^4}\right)\Bigg\}\Bigg] = o(1),
\end{align*}
where $\kappa(X)$ is the condition number of $X$, and $m_{\max} = \max\{m_1, \dots, m_n\}$. 
\end{theorem}

\begin{proof}
We begin with decomposing the squared loss with respect to the posterior mean $\bar \beta (\phi)$ as follows:
    \begin{align}\label{eq:mse}
  \| X \beta_0 - X \bar \beta (\phi) \|^2 &=  \big \|  X \beta_0 -  X \{  X^{\T} C^{-1} X + (\delta^{2} \Lambda)^{-1}\}^{-1} X^{\T} C^{-1} y \big \|_2^2 \nonumber \\
    &=
    \big \|  X \beta_0 -  X \{  X^{\T} C^{-1} X + (\delta^{2} \Lambda)^{-1}\}^{-1} X^{\T} C^{-1} (X \beta_0 + \epsilon_0) \big \|_2^2 \nonumber  \\
    &\leq 2 \big \|  \left[ X - X \{  X^{\T} C^{-1} X + (\delta^{2} \Lambda)^{-1}\}^{-1} X^{\T} C^{-1} X \right] \beta_0 \big \|_2^2 + \nonumber \\
    &\quad \; 2
    \big \|X \{  X^{\T} C^{-1} X + (\delta^{2} \Lambda)^{-1}\}^{-1} X^{\T} C^{-1} \epsilon_0 \big \|_2^2.
\end{align}
Lemma \ref{auxlem1} implies that \eqref{eq:mse} reduces to
\begin{align}
  \label{eq:mse1}
   \| X \beta_0 - X \bar \beta (\phi) \|_2^2 &\leq 2 \| C (X \delta^2 \Lambda X^\T
   + C)^{-1} X \beta_0 \|_2^{2} + 2 \| X\delta^2 \Lambda X^\T (X \delta^2 \Lambda
   X^\T + C)^{-1} \epsilon_0 \|_2^2 \nonumber \\
   &\equiv 2 T_1 + 2 T_2.
\end{align}
Lemma \ref{lemt1} implies that $T_1$ in \eqref{eq:mse1} satisfies
\begin{equation}\label{boundEt1}
    2\EE(T_1) \leq 2\EE_{X} \left\{
\kappa^4(X) \right\} \| \beta_{0}\|_2^{2}\left\{O\left(\frac{nb^8\sigma_Z^8 q^{12}}{k^8} + \frac{nb^8\sigma_Z^8 m_{\max}^4}{k^8} + nb^8\sigma_Z^8 \right)\right\},
\end{equation}
where $\kappa(X) = s_{\max}(X) / s_{\min}(X)$ is the condition number, and $s_{\max}(X)$ and $s_{\min}(X)$ are the maximum and minimum singular values of $X$. Lemma \ref{lemt2} implies that $T_2$ in \eqref{eq:mse1} satisfies
\begin{equation}\label{boundEt2}
    2\EE (T_2) \leq \frac{4\tau_{0}^{2}}{a^{4}} \; \EE_{X}\bigl\{\kappa^4(X)\bigr\} \; \left\{ \left(\frac{N-p}{\log N}\right)O\left( \frac{nb^5\sigma_Z^6 (q^9 + m_{\max}^3)}{k^6}\right) + (p + 4e^{-p/8})O\left(\frac{nb^3\sigma_Z^4(q^6+m_{\max}^2)}{k^4}\right)\right\},
\end{equation}
where $a$ and $b$ are constants defined in Assumptions (A1) and (A3) respectively.

Combining \eqref{boundEt1} and \eqref{boundEt2}, we have the prediction risk
\begin{align}\label{eq:mse-final}
    \frac{1}{N}\EE \| X \beta_0 - X \bar \beta (\phi) \|_2^2 &\leq  \EE_{X}\left\{ \kappa^4(X)\right\} \Bigg[ \frac{2\| \beta_{0}\|_2^{2}}{N}\left\{O\left(\frac{nb^8\sigma_Z^8 q^{12}}{k^8} + \frac{nb^8\sigma_Z^8 m_{\max}^4}{k^8} + nb^8\sigma_Z^8 \right)\right\} \nonumber \\
    &\qquad \qquad \qquad \qquad + \frac{4\tau_{0}^{2}}{a^{4}} \Bigg\{ \left(\frac{N-p}{N\log N}\right)O\left( \frac{nb^5\sigma_Z^6 (q^9 + m_{\max}^3)}{k^6}\right) \nonumber\\
        &\qquad \qquad \qquad \qquad + \left(\frac{p + 4e^{-p/8}}{N}\right)O\left(\frac{nb^3\sigma_Z^4(q^6+m_{\max}^2)}{k^4}\right)\Bigg\}\Bigg] \nonumber \\
    &= o(1),
\end{align}
since under Assumption (A4), $\EE_{X} \left\{ \kappa^4(X) \right\} = 1 + o(1)$; see \cite{KSS:2025} for details. This establishes the required result. \end{proof}

\subsection{Upper Bound for $T_2$}\label{Section:ProofT2}

\begin{lemma}\label{lemt2}
Under the assumptions of Theorem 3.1 in the main manuscript,
\begin{align*}
\EE (T_2) \leq \frac{2\tau_{0}^{2}}{a^{4}} \; \EE_{X}\bigl\{\kappa^4(X)\bigr\} \; \left\{ \left(\frac{N-p}{\log N}\right)O\left( \frac{nb^5\sigma_Z^6 (q^9 + m_{\max}^3)}{k^6}\right) + (p + 4e^{-p/8})O\left(\frac{nb^3\sigma_Z^4(q^6+m_{\max}^2)}{k^4}\right)\right\}
\end{align*}
where $\kappa(X)$ is the condition number of $X$ and $m_{\max} = \max\{m_1, \dots, m_n\}$.
\end{lemma}

\begin{proof}
The term $T_{2} = \| X\delta^2 \Lambda X^\T (X \delta^2 \Lambda
   X^\T + C)^{-1} \epsilon_0 \|_2^2$ depends on $X\delta^2 \Lambda X^\T$ and $X \delta^2 \Lambda
   X^\T + C$. Let the singular value decomposition of $X$ and
spectral decomposition of $X \delta^2 \Lambda X^\T$ be
\begin{align} \label{eq:1}
  X = U
  \begin{bmatrix}
  D \\
  0_{N-p \times p}
  \end{bmatrix}
  V^\T, \quad
  X \delta^2 \Lambda X^\T = U
\begin{bmatrix}
  D V^\T \delta^2 \Lambda V D & 0_{p \times N - p} \\
  0_{N-p \times p} & 0_{N - p \times N - p}
  \end{bmatrix}
  U^\T,
\end{align}
where $N > p$, $U$ and $V$ are $N \times N$ and $p \times p$
orthonormal matrices, and $D$ is a $p \times p$ diagonal matrix of singular
values of $X$. If $\tilde D = D V^\T \delta^2 \Lambda V D$ and $\tilde C = U^{\T} C U$, then the  second
decomposition in (\ref{eq:1}) implies that
\begin{align}\label{eq:2}
  X \delta^2 \Lambda X^\T
+ C =   X \delta^2 \Lambda X^\T
+ U U^{\T} C U U^{\T} = U
\begin{bmatrix}
  \tilde D + \tilde C_{11} & \tilde C_{12} \\
  \tilde C_{12}^{\T} & \tilde C_{22}
  \end{bmatrix}
  U^\T,
\end{align}
where $\tilde C_{ij}$ is the $ij$-th block of the $2 \times 2$ block matrix
$\tilde C$. Substituting (\ref{eq:1}) and (\ref{eq:2}) in $T_{2}$ implies that
\begin{align}
\label{eq:3}
T_{2} &=
\left \|  U
\begin{bmatrix}
  \tilde D & 0_{p \times N - p} \\
  0_{N-p \times p} & 0_{N - p \times N - p}
  \end{bmatrix}
  U^\T U
\begin{bmatrix}
  \tilde D + \tilde C_{11} & \tilde C_{12} \\
  \tilde C_{12}^{\T} & \tilde C_{22}
  \end{bmatrix}^{{-1}}
  U^\T \epsilon_0 \right  \|_2^{2} \nonumber \\
  &=
\left \|
\begin{bmatrix}
  \tilde D & 0_{p \times N - p} \\
  0_{N-p \times p} & 0_{N - p \times N - p}
  \end{bmatrix}
\begin{bmatrix}
  \tilde D + \tilde C_{11} & \tilde C_{12} \\
  \tilde C_{12}^{\T} & \tilde C_{22}
  \end{bmatrix}^{{-1}}
  \tilde \epsilon \right  \|_2^{2},
\end{align}
where the first $U$ leaves the norm unchanged because $U^{\T} U = I_{N}$ and
$\tilde \epsilon = U^{\T} \epsilon_0$ is distributed as $\Ncal(0, \tau_{0}^{2}
I_{N})$ because $U$ is orthonormal. Furthermore, $C$ is positive definite because $C=\text{diag}(C_1, \dots, C_n)$ with $C_i(\Gamma) = Z_i S^{\T}
\Gamma R  \left( Z_i S^{\T} \Gamma R \right)^{\T} + I_{m_i}$, which implies that $\tilde C = U^{\T}C U$ is a positive definite matrix because $U$ is orthonormal. Hence, $\tilde C_{22}^{-1}$ exists. If $F = \tilde D + \tilde C_{11} - \tilde
C_{12} \tilde C_{22}^{-1} \tilde C_{12}^{\T}$, then the block matrix inversion
formula \citep{Har97} gives
\begin{align}\label{eq:blockmatinv}
&\begin{bmatrix}
  \tilde D + \tilde C_{11} & \tilde C_{12} \\
  \tilde C_{12}^{\T} & \tilde C_{22}
  \end{bmatrix}^{{-1}} =
\begin{bmatrix}
  F^{-1} &  - F^{{-1}} \tilde C_{12} \tilde C_{22}^{-1} \\
  - \tilde C_{22}^{{-1}} \tilde C_{12}^{\T} F^{-1}  & \tilde C_{22}^{-1} + \tilde
  C_{22}^{-1} C_{12}^{\T} F^{{-1}} \tilde C_{12} \tilde C_{22}^{-1}
\end{bmatrix}.
\end{align}
Substituting this identity in (\ref{eq:3}) implies that
\begin{align}\label{eq:22}
\begin{bmatrix}
  \tilde D & 0_{p \times N - p} \\
  0_{N-p \times p} & 0_{N - p \times N - p}
  \end{bmatrix}
&\begin{bmatrix}
  \tilde D + \tilde C_{11} & \tilde C_{12} \\
  \tilde C_{12}^{\T} & \tilde C_{22}
  \end{bmatrix}^{{-1}} =
\begin{bmatrix}
  \tilde D F^{-1} &  - \tilde D  F^{{-1}} \tilde C_{12} \tilde C_{22}^{-1} \\
   0_{N-p \times p} & 0_{N - p \times N - p}
\end{bmatrix}.
\end{align}
Using (\ref{eq:22}), (\ref{eq:3}) is written as the quadratic form
\begin{align}
\label{eq:5}
  T_{2} &= \tilde \epsilon^{\T}
  \begin{bmatrix}
   F^{-1} \tilde D &  0_{p \times N- p} \\
   - \tilde C_{22}^{-1} \tilde C_{12}^{\T} F^{{-1}} \tilde D      & 0_{N - p \times N - p}
  \end{bmatrix}
  \begin{bmatrix}
  \tilde D F^{-1} &  - \tilde D  F^{{-1}} \tilde C_{12} \tilde C_{22}^{-1} \\
   0_{N-p \times p} & 0_{N - p \times N - p}
\end{bmatrix}  \tilde \epsilon \nonumber \\
&= \tilde \epsilon^{\T}
  \begin{bmatrix}
   F^{-1} \tilde D^{2} F^{-1} &  - F^{{-1}} \tilde D^{2}  F^{{-1}} \tilde C_{12} \tilde C_{22}^{-1} \\
   - \tilde C_{22}^{-1} \tilde C_{12}^{\T} F^{{-1}} \tilde D^{2} F^{-1}      &
   \tilde C_{22}^{-1} \tilde C_{12}^{\T} F^{{-1}} \tilde D^{2} F^{-1} \tilde C_{12} \tilde C_{22}^{-1}
  \end{bmatrix}
  \tilde \epsilon,
\end{align}
which is further simplified by partitioning $\tilde \epsilon$ into two blocks as $\tilde \epsilon^{\T} = (\tilde \epsilon_{1}^{\T}, \tilde
\epsilon_{2}^{\T})$ of dimensions $p \times 1$ and $(N -
p) \times 1$, respectively.  Using Cauchy-Schwartz inequality for quadratic forms in
\eqref{eq:5} gives
\begin{align}
\label{eq:6}
  T_{2} \leq  2 \tilde \epsilon_{1}^{\T} F^{-1} \tilde D^{2} F^{-1} \tilde
  \epsilon_{1} + 2 \tilde \epsilon_{2}^{\T} \tilde C_{22}^{-1} \tilde C_{12}^{\T}
  F^{{-1}} \tilde D^{2} F^{-1} \tilde C_{12} \tilde C_{22}^{-1} \tilde \epsilon_{2}.
\end{align}
The first term in (\ref{eq:6}) satisfies
\begin{align}
\label{eq:7}
\tilde \epsilon_{1}^{\T} F^{-1} \tilde D^{2} F^{-1} \tilde \epsilon_{1}
\leq  \lambda_{\max}(\tilde D^{2}) \;   \tilde \epsilon_{1}^{\T} F^{-1} F^{-1}
\tilde \epsilon_{1} = \lambda^{2}_{\max}(\tilde D) \;   \tilde \epsilon_{1}^{\T} F^{-2} \tilde \epsilon_{1},
\end{align}
where $\lambda_{\max}(\tilde D)$ is the maximum eigen value of $\tilde D$ and $\lambda_{\max}(\tilde D^{2}) = \lambda^{2}_{\max}(\tilde D)$ because $\tilde D$ is a symmetric matrix. 
The Schur complement of $\tilde C_{22}$ in $\tilde C$ is $\tilde C_{11} - \tilde C_{12}
\tilde C_{22}^{-1} \tilde C_{12}^{\T}$, so $\lambda_{\min}(\tilde C_{11} - \tilde C_{12}
\tilde C_{22}^{-1} \tilde C_{12}^{\T}) \geq 0$, where $\lambda_{\min}(A)$ is the minimum eigen value of $A$. This implies that
\begin{align}
\label{eq:9}
  F^{{-1}} = (\tilde D + \tilde C_{11} - \tilde C_{12} \tilde C_{22}^{-1} \tilde
C_{12}^{\T})^{{-1}} \preceq \tilde D^{-1} \preceq \lambda_{\min}^{-1}(\tilde D) \; I_{p},
  \quad F^{{-2}} \preceq \lambda_{\min}^{-2}(\tilde D) \; I_{p},
\end{align}
where $A \preceq B$ for two positive semi-definite matrices $A$ and $B$ implies
that $\lambda_{\min}(B - A) \geq 0$. Using \eqref{eq:9} in \eqref{eq:7} gives
\begin{align}
\label{eq:10}
\tilde \epsilon_{1}^{\T} F^{-1} \tilde D^{2} F^{-1} \tilde \epsilon_{1}
\leq  \frac{\lambda_{\max}^{2}(\tilde D)}{\lambda_{\min}^{2}(\tilde D)} \;   \tilde \epsilon_{1}^{\T} \tilde \epsilon_{1}.
\end{align}

We simplify $\tilde D$ to find upper and lower bounds for
$\lambda_{\max}(\tilde D)$ and $\lambda_{\min}(\tilde D)$. The Assumption (A1) in the main manuscript on the support of $\Lambda$ implies that $\lambda_{j}^2 \in
[a, 1/a]$ for $j=1, \ldots, p$, so 
\begin{align}
\label{eq:4}
\tilde D &= D V^\T \delta^2 \Lambda V D \preceq
a^{-1} \delta^2  D V^\T  V D  = a^{-1} \delta^2
D^{2}, \nonumber\\
 \tilde D &= D V^\T \delta^2 \Lambda V D \succeq
 a \delta^2  D V^\T  V D = a \delta^2 D^{2}.
\end{align}
Furthermore, (\ref{eq:1}) implies that $D_{11} = s_{\max}(X)$ and $D_{pp} =
s_{\min}(X)$, where $s_{\max}(X)$ and $s_{\min}(X)$ are the maximum and minimum singular values of $X$. Using this in \eqref{eq:4} gives
\begin{align}
\label{eq:23}
\lambda_{\max}(\tilde D) \leq  a^{-1} \delta^{2} s^{2}_{\max}(X), \quad
\lambda_{\min}(\tilde D) \geq a \delta^{2}  s^{2}_{\min}(X), \quad
\frac{\lambda_{\max}^{2}(\tilde D)}{\lambda_{\min}^{2}(\tilde D)} \leq \frac{1}{a^{4}} \frac{s_{\max}^{4}(X)}{s_{\min}^{4}(X)}.
\end{align}
Substituting the last inequality in \eqref{eq:10} implies that
\begin{align}
\label{eq:24}
\tilde \epsilon_{1}^{\T} F^{-1} \tilde D^{2} F^{-1} \tilde \epsilon_{1}
\leq  \frac{1}{a^{4}} \; \kappa^4(X)  \; \tilde \epsilon_{1}^{\T} \tilde \epsilon_{1}, 
\end{align}
where $\kappa(X) = {s_{\max}(X)}/{s_{\min}(X)}$ denotes the condition number of $X$. Finally, we take expectations on both sides in (\ref{eq:24}). The term $\tilde
\epsilon_{1}^{\T} \tilde \epsilon_{1}$ only depends on the distribution of
$\tilde \epsilon$ and is independent of the distribution of $X$, implying that
\begin{align}
\label{eq:16}
\EE (\tilde \epsilon_{1}^{\T} F^{-1} \tilde D^{2} F^{-1} \tilde \epsilon_{1})
\leq  \frac{1}{a^{4}} \; \EE_{X} \{\kappa^4(X)\}  \;  \EE_{\tilde \epsilon} (\tilde
\epsilon_{1}^{\T} \tilde \epsilon_{1}) = \frac{\tau_{0}^{2}}{a^{4}} \; \EE_{X} \{\kappa^4(X)\}  \;  p,
\end{align}
where $\EE_{\tilde \epsilon} (\tilde \epsilon_{1}^{\T} \tilde \epsilon_{1}) = \tau_{0}^2 p$
because $\epsilon_{1} / \tau_{0}$ is distributed as $\Ncal(0, I_{p})$.

Consider the second term in (\ref{eq:6}). The bounds in (\ref{eq:10}) and
(\ref{eq:23}) imply that
\begin{align}
\label{eq:11}
\tilde \epsilon_{2}^{\T} \tilde C_{22}^{-1} \tilde C_{12}^{\T} F^{{-1}} \tilde
D^{2} F^{-1} \tilde C_{12} \tilde C_{22}^{-1} \tilde \epsilon_{2} &\leq
\frac{\lambda_{\max}^{2}(\tilde D)}{\lambda_{\min}^{2}(\tilde D)} \;   \tilde
\epsilon_{2}^{\T} \tilde C_{22}^{-1} \tilde C_{12}^{\T}  \tilde C_{12} \tilde
C_{22}^{-1} \tilde \epsilon_{2} \nonumber \\
&\leq \frac{1}{a^{4}} \kappa^4(X) \;   \tilde
\epsilon_{2}^{\T} \tilde C_{22}^{-1} \tilde C_{12}^{\T}  \tilde C_{12} \tilde
C_{22}^{-1} \tilde \epsilon_{2}.
\end{align}
The quadratic form $\tilde
\epsilon_{2}^{\T} \tilde C_{22}^{-1} \tilde C_{12}^{\T}  \tilde C_{12} \tilde
C_{22}^{-1} \tilde \epsilon_{2}$ is independent of the distribution of $X$, so taking expectation on both sides of (\ref{eq:11}) gives
\begin{align}
\label{eq:12}
    \EE (\tilde \epsilon_{2}^{\T} \tilde C_{22}^{-1} \tilde C_{12}^{\T} F^{{-1}} \tilde
D^{2} F^{-1} \tilde C_{12} \tilde C_{22}^{-1} \tilde \epsilon_{2}) &\leq \frac{1}{a^4}\EE_{X}\left\{\kappa^4(X) \right\} \EE(\tilde
\epsilon_{2}^{\T} \tilde C_{22}^{-1} \tilde C_{12}^{\T}  \tilde C_{12} \tilde
C_{22}^{-1} \tilde \epsilon_2).
\end{align}
If $\EE_{R,S, Z}, \EE_{\tilde \epsilon \mid R, S, Z}$, and $\EE_{\Gamma \mid y}$ denote the expectations with respect to the distributions of $(R, S, Z)$ and $\tilde \epsilon$ given $(R, S, Z)$ and the posterior distribution of $\Gamma $ given $y$, then the second expectation on the right of \eqref{eq:12} is
\begin{align}
\label{eq:13}
\EE(\tilde \epsilon_{2}^{\T} \tilde C_{22}^{-1} \tilde C_{12}^{\T}  \tilde
C_{12} \tilde C_{22}^{-1} \tilde \epsilon_{2}) &=\EE_{R,S,Z} \EE_{\tilde \epsilon \mid R,S,Z} \EE_{\Gamma \mid y}
(\tilde \epsilon_{2}^{\T} \tilde C_{22}^{-1} \tilde C_{12}^{\T}  \tilde C_{12}
\tilde C_{22}^{-1} \tilde \epsilon_{2})\nonumber \\
&\leq \EE_{R,S,Z}\EE_{\tilde \epsilon \mid R,S,Z} \{\sup_{{\Gamma: \| \Gamma \| \leq b}} (\tilde
\epsilon_{2}^{\T} \tilde C_{22}^{-1} \tilde C_{12}^{\T}  \tilde C_{12} \tilde
C_{22}^{-1} \tilde \epsilon_{2}) \EE_{\Gamma \mid y} (1) \} \nonumber \\
& =  \EE_{R,S,Z}\EE_{\tilde \epsilon \mid R,S,Z} \{\sup_{{\Gamma: \| \Gamma \| \leq b}} (\tilde
\epsilon_{2}^{\T} \tilde C_{22}^{-1} \tilde C_{12}^{\T}  \tilde C_{12} \tilde
C_{22}^{-1} \tilde \epsilon_{2}) \},
\end{align}
where $b$ is a positive constant. 

We use a covering argument to bound the last term involving the supremum. Lemma \ref{auxlem2} shows that there is a $1/ \log N$-net, denoted as $\tilde \Gcal$, for the set $\Gcal = \{ \Gamma \in \RR^{k \times k}: \| \Gamma \| \leq b\}$ such that the cardinality of $\tilde \Gcal$ is $\lceil 2bk\log N \rceil^{k^2}$. This implies that for any $\Gamma \in \Gcal$, there is a  $\tilde \Gamma \in \tilde \Gcal$ such that $\| \Gamma - \tilde \Gamma \| \leq 1 / \log N$. Let  $A(\Gamma) = \tilde C_{12}(\Gamma) \tilde C_{22} (\Gamma)^{-1}$ be the matrix that appears in \eqref{eq:13}, where the notation emphasizes the dependence on $\Gamma$.
We show that the norms of terms involving $A(\Gamma)$ for any $\Gamma \in \Gcal$ are controlled using $A(\tilde \Gamma)$ for  $\tilde \Gamma \in \tilde \Gcal$. Specifically,
\begin{align}\label{eq:lipsch1}
\tilde \epsilon_{2}^{\T} \left\{ A(\Gamma)^{\T} A(\Gamma) - A(\tilde \Gamma)^{\T} A(\tilde \Gamma)\right\} \tilde \epsilon_{2} &\leq \| A(\Gamma)^{\T} A(\Gamma) - A(\tilde \Gamma)^{\T} A(\tilde \Gamma) \| \| \tilde \epsilon_{2} \|^2 \nonumber\\
&\overset{(i)}{\leq} 4b (2\breve c + 3 b^2 \breve c^2 + b^4 \breve c^3)\, \| \Gamma - \tilde \Gamma\| \nonumber\\
&\overset{(ii)}{\leq} 4b (2\breve c + 3 b^2 \breve c^2 + b^4 \breve c^3) /\log N
\end{align} 
where $(i)$ follows from Lemma \ref{auxlem5}, $\breve c$ is defined in Lemma \ref{auxlem5},  and $(ii)$ follows because $\tilde \Gcal$ is a $1/\log N$-net of $\Gcal$.

Using  \eqref{eq:lipsch1} in  \eqref{eq:13} gives 
\begin{align}
\label{eq:28}
\sup_{{\Gamma: \| \Gamma \| \leq b}} (\tilde
\epsilon_{2}^{\T} \tilde C_{22}^{-1} \tilde C_{12}^{\T}  \tilde C_{12} \tilde
C_{22}^{-1} \tilde \epsilon_{2}) &=
\sup_{{\Gamma: \| \Gamma \| \leq b}} \tilde \epsilon_{2}^{\T} \left[
 \left\{ A(\Gamma)^{\T} A(\Gamma) - A(\tilde \Gamma)^{\T} A(\tilde \Gamma)\right\} + \left\{ A(\tilde \Gamma)^{\T} A(\tilde \Gamma)\right\} \right] \tilde \epsilon_{2} \nonumber\\
&\leq \frac{4b}{\log N} (2\breve c + 3 b^2 \breve c^2 + b^4 \breve c^3)  \| \tilde \epsilon_{2} \|^{2} + \underset{\tilde \Gamma \in \tilde \Gcal}{\max} \; \tilde \epsilon_{2}^{\T} \left\{ A(\tilde \Gamma)^{\T} A(\tilde \Gamma)\right\} \tilde \epsilon_{2},\nonumber\\
\EE_{\tilde \epsilon \mid R,S, Z} \left\{ \sup_{{\Gamma: \| \Gamma \| \leq b}} (\tilde
\epsilon_{2}^{\T} \tilde C_{22}^{-1} \tilde C_{12}^{\T}  \tilde C_{12} \tilde
C_{22}^{-1} \tilde \epsilon_{2}) \right\} &\leq   \frac{4b}{\log N} (2\breve c + 3 b^2 \breve c^2 + b^4 \breve c^3)   
 \EE_{\tilde \epsilon \mid R,S, Z} \| \tilde \epsilon_{2} \|^{2} + \nonumber\\
&\qquad \EE_{\tilde \epsilon \mid R, S, Z} \underset{\tilde \Gamma \in \tilde \Gcal}{\max} \; \tilde \epsilon_{2}^{\T} \left\{ A(\tilde \Gamma)^{\T} A(\tilde \Gamma)\right\} \tilde \epsilon_{2} \nonumber\\
&\overset{(i)}{=} \frac{4b\tau_0^2(N-p)}{\log N} (2\breve c + 3 b^2 \breve c^2 + b^4 \breve c^3)  
    + \nonumber\\
&\qquad \EE_{\tilde \epsilon \mid R, S, Z} \underset{\tilde \Gamma \in \tilde \Gcal}{\max} \; \tilde \epsilon_{2}^{\T} \left\{ A(\tilde \Gamma)^{\T} A(\tilde \Gamma)\right\} \tilde \epsilon_{2},
\end{align}
where $(i)$ follows because $ \epsilon_{2} / \tau_{0}$ is distributed as $\Ncal(0, I_{N-p})$. 
For every $\tilde \Gamma \in \tilde \Gcal$, the rank of $A(\tilde \Gamma)  = \tilde C_{12} (\tilde \Gamma) \tilde C_{22}(\tilde \Gamma)^{-1}$ is $\min\{k, p\}$ and \eqref{eq:bound_A_Gamma} implies that $\| A(\tilde \Gamma)\|$ is bounded above by $c_*$, where $c_* = 1 + b^2 \check c$ is the upper bound on $\|C\|$. This in turn implies that $\|A(\tilde \Gamma)\|^{2}_{2}$ is stochastically dominated by $c_*^{2} \tau_{0}^{2}\| \breve{\epsilon} \|^{2}$, where $\breve{\epsilon}$ is distributed as $\Ncal(0, I_{\min(k, p)})$ and $\| \breve{\epsilon} \|^{2}$ is distributed as $\chi^{2}_{\min(k, p)}$. Using this fact, we provide an upper bound for the expectation in \eqref{eq:28} as follows:
\begin{align}
\label{eq:31}
\EE_{\tilde \epsilon \mid R, S, Z} \underset{\tilde \Gamma \in \tilde \Gcal}{\max} \; \tilde \epsilon_{2}^{\T} A(\tilde \Gamma)^{\T} A(\tilde \Gamma) \tilde \epsilon_{2} & = \int_{0}^{\infty} \PP \left \{ \underset{\tilde \Gamma \in \tilde \Gcal}{\max} \; \| A(\tilde \Gamma) \tilde \epsilon_{2} \|^{2}_{2} > t \right\} dt \nonumber \\
&= 2c_*^{2} \tau_{0}^{2} p  +
\int_{2 c_*^{2} \tau_{0}^{2} p}^{\infty} \PP \left \{ \underset{\tilde \Gamma \in \tilde \Gcal}{\max} \;\| A(\tilde \Gamma) \tilde \epsilon_{2} \|^{2}_{2} > t \right\} dt, \nonumber \\
\int_{2 c_*^{2} \tau_{0}^{2} p}^{\infty} \PP \left \{ \underset{\tilde \Gamma \in \tilde \Gcal}{\max} \;\| A(\tilde \Gamma) \tilde \epsilon_{2} \|^{2}_{2} > t \right\} dt
 &\overset{(i)}{\leq}
\sum_{\tilde \Gamma \in \tilde \Gcal} \int_{2 c_*^{2} \tau_{0}^{2} p}^{\infty} \PP \left \{ \| A(\tilde \Gamma) \tilde \epsilon_{2} \|^{2}_{2} > t \right\} dt \nonumber \\
&\overset{(ii)}{\leq}
\sum_{\tilde \Gamma \in \tilde \Gcal}
\int_{2 c_*^{2} \tau_{0}^{2} p}^{\infty} \PP \left \{ c_*^{2} \tau_{0}^{2} \chi^{2}_{\min(k, p)} > t \right\} dt \nonumber \\
&\overset{(iii)}{\leq}
\sum_{\tilde \Gamma \in \tilde \Gcal}
\int_{2 c_*^{2} \tau_{0}^{2} p}^{\infty} \PP \left( \chi^{2}_{p} > t / c_*^{2} \tau_{0}^{2}  \right) dt \nonumber\\
&\overset{(iv)}{=}
c_*^{2} \tau_{0}^{2}  \lceil 2 b k \log N \rceil^{k^2}
\int_{p}^{\infty} \PP \left( \chi^{2}_{p} - p > u  \right) du,
\end{align}
where $(i)$ follows from the union bound, $(ii)$ follows because  $c_*^{2} \tau_{0}^{2}\| \breve{\epsilon} \|^{2}$ stochastically dominates $\| A(\tilde \Gamma) \tilde \epsilon_{2} \|^{2}_{2}$,  $(iii)$ follows because $ \chi^{2}_{p}$ stochastically dominates $\chi^{2}_{\min(k, p)}$, and $(iv)$ follows from the substitution $u = t/c_*^{2} \tau_{0}^{2} - p$. Definition 2.7 in \citet{Wainwright_2019} implies that if $f$ follows $\chi^{2}_{p}$ distribution, then its moment generating function is
\begin{align*}
\EE \{e^{\lambda (f - p)}\} = \frac{e^{- \lambda p}}{(1 - 2 \lambda)^{p/2}} \leq e^{4 p \lambda^{2} / 2}, \quad |\lambda | < 1/4,
\end{align*}
and it is subexponential with parameters $(\nu, \alpha) = (2 \sqrt{p}, 4)$. Proposition 2.9  in \citet{Wainwright_2019} gives
\begin{align}
\label{eq:33}
\int_{p}^{\infty} \PP \left( \chi^{2}_{p} - p > u  \right) du \leq \int_{p}^{\infty} e^{- \frac{u}{8}} du = 8 e^{-p/8}.
\end{align}
Substituting (\ref{eq:33}) in (\ref{eq:31}) gives
\begin{align}
\label{eq:131}
\EE_{\tilde \epsilon \mid R, S, Z} \underset{\tilde \Gamma \in \tilde \Gcal}{\max} \; \tilde \epsilon_{2}^{\T} A(\tilde \Gamma)^{\T} A(\tilde \Gamma) \tilde \epsilon_{2}
&\leq 2 c_*^{2} \tau_{0}^{2} p  + 8 c_*^{2} \tau_{0}^{2}  \lceil 2 b k \log N \rceil^{k^2}  e^{-p/8} \nonumber \\
&= 2 c_*^{2} \tau_{0}^{2} p  + 8 c_*^{2} \tau_{0}^{2} e^{k^{2} \log \lceil 2 b k \log N \rceil -p/8 } \nonumber\\
&\leq 2 c_*^{2} \tau_{0}^{2} p  + 8 c_*^{2} \tau_{0}^{2} e^{k^{2} \log  \{(2 b + 1) k \log N\}  -p/8 } \nonumber\\
&= 2 c_*^{2} \tau_{0}^{2} p  + 8 c_*^{2} \tau_{0}^{2} e^{k^{2} \log  (2 b + 1) +  k^{2} \log k + k^{2} \log \log N  -p/8 } \nonumber\\
&\overset{(i)}{=} 2 c_*^{2} \tau_{0}^{2} p  + 8 c_*^{2} \tau_{0}^{2}e^{-p/8 \{1 + o(1) \}}, 
\end{align}
where $(i)$ follows if  $k^{2} \log k = o(p)$ and $k^{2} \log \log N = o(p)$.

We combine \eqref{eq:13}, \eqref{eq:28}, \eqref{eq:131} and use them in \eqref{eq:12} to obtain
\begin{align}
\label{eq:34}
&\EE (\tilde \epsilon_{2}^{\T} \tilde C_{22}^{-1} \tilde C_{12}^{\T} F^{{-1}} \tilde
D^{2} F^{-1} \tilde C_{12} \tilde C_{22}^{-1} \tilde \epsilon_{2}) \nonumber\\
&\leq \frac{\tau_0^2}{a^4}\EE_{X}\left\{ \kappa^4(X) \right\} \, \EE_{R, S, Z} \left[\frac{4b(N-p)}{\log N} (2\breve c + 3 b^2 \breve c^2 + b^4 \breve c^3)  + 2 c_*^{2} p  + 8c_*^{2} e^{-p/8 }\right] \nonumber\\
&= \frac{2\tau_0^2}{a^4}\EE_{X}\left\{ \kappa^4(X) \right\} \, \EE_{R, S, Z} \left[\frac{2b(N-p)}{\log N} (2\breve c + 3 b^2 \breve c^2 + b^4 \breve c^3)  +  (1 + b^4 \breve c^2 + 2 b^2 \breve c) (p  + 4 e^{-p/8}) \right] \nonumber\\
&= \frac{2\tau_0^2}{a^4}\EE_{X}\left\{ \kappa^4(X) \right\} \, \left\{ (p  + 4 e^{-p/8}) + a_{1} \EE_{R, S, Z}(\breve c) + a_2 \EE_{R, S, Z} (\breve c^2) + a_3 \EE_{R, S, Z} (\breve c^3) \right\}, \nonumber\\
a_1 &= \frac{4b(N-p)}{\log N} + 2 b^2 (p  + 4 e^{-p/8}), \quad a_2 = \frac{6b^3(N-p)}{\log N} + b^4(p  + 4 e^{-p/8}), \quad 
a_3 =  \frac{2b^5(N-p)}{\log N}. 
\end{align}
Finally, taking expectation on both sides of \eqref{eq:6}, and substituting \eqref{eq:16}, \eqref{eq:34}, 
and using \eqref{anrs-bd-2} for $r=4$ and $r=6$ in Lemma \ref{auxlemma7}
gives
\begin{align*}
\EE (T_2)
&\leq \frac{2\tau_{0}^{2}}{a^{4}} \; \EE_{X}\bigl\{\kappa^4(X)\bigr\} \; \Bigl[3p + 8e^{-p/8} + \frac{2^8  b\sigma_Z^2}{k^2}\Bigl\{\frac{2(N-p)}{\log N} + b (p  + 4 e^{-p/8})\Bigr\} (q+k+2)^2\,\{n(q+2)+N\} \\
&\quad + \frac{2^{19}b^3\sigma_Z^4}{k^4}\, \Bigl\{\frac{6(N-p)}{\log N}  + b(p + 4e^{-p/8}) \Bigr\}(q^2+k^2+2)^2\,\Bigl\{n(q^2+2)+\sum_{i=1}^n m_i^2\Bigr\} \\
&\quad + \frac{2^{32}b^5\sigma_Z^6}{k^6}\,\Bigl( \frac{N-p}{\log N} \Bigr)(q^3+k^3+3)^2\,\Bigl\{n(q^3+3)+\sum_{i=1}^n m_i^3\Bigr\}\Bigr] \nonumber \\
&\overset{(i)}{\leq} \frac{2\tau_{0}^{2}}{a^{4}} \; \EE_{X}\bigl\{\kappa^4(X)\bigr\} \; \Bigl[3p + 8e^{-p/8} + \frac{2^8  b\sigma_Z^2}{k^2}\Bigl\{\frac{2(N-p)}{\log N} + b (p  + 4 e^{-p/8})\Bigr\} (4q^2+4k^2+8)\,\{n(q+2)+N\} \\
&\quad + \frac{2^{19}b^3\sigma_Z^4}{k^4}\, \Bigl\{\frac{6(N-p)}{\log N}  + b(p + 4e^{-p/8}) \Bigr\}(4q^4+4k^4+8)\,\Bigl\{n(q^2+2)+n m_{\max}^2\Bigr\} \\
&\quad + \frac{2^{32}b^5\sigma_Z^6}{k^6}\,\Bigl( \frac{N-p}{\log N} \Bigr)(4q^6+4k^6+18)\,\Bigl\{n(q^3+3)+n m_{\max}^3\Bigr\}\Bigr], \\
&= \frac{2\tau_{0}^{2}}{a^{4}} \; \EE_{X}\bigl\{\kappa^4(X)\bigr\} \; \left\{ \left(\frac{N-p}{\log N}\right)O\left( \frac{nb^5\sigma_Z^6 (q^9 + m_{\max}^3)}{k^6}\right) + (p + 4e^{-p/8})O\left(\frac{nb^3\sigma_Z^4(q^6+m_{\max}^2)}{k^4}\right)\right\}
\end{align*}
where $(i)$ follows from the inequality $(u+v)^2 \leq 2(u^2 + v^2)$. The lemma is proved.
\end{proof}

\subsection{Upper Bound for $T_1$}\label{Section:ProofT1}
\begin{lemma}\label{lemt1}
Under the assumptions of Theorem 3.1 in the main manuscript,
\begin{align*}
\EE(T_1) \leq \EE_{X} \left\{
\kappa^4(X) \right\} \| \beta_{0}\|_2^{2} O\left(\frac{nb^8\sigma_Z^8 q^{12}}{k^8} + \frac{nb^8\sigma_Z^8 m_{\max}^4}{k^8} + nb^8\sigma_Z^8 \right) 
\end{align*}
where $\kappa(X)$ is the condition number of $X$ and $m_{\max} = \max\{m_1, \dots, m_n\}$.
\end{lemma}
\begin{proof}
The term $T_{1}$ satisfies
\begin{align}
\label{eq:17}
T_1 \leq \| C \|^2 \, \| (X \delta^2 \Lambda X^\T + C)^{-1} X  \beta_{0}\|_2^2 \leq c_*^2 \, \| (X \delta^2 \Lambda X^\T + C)^{-1} X  \beta_{0}\|_2^2 ,
\end{align}
where $c_* = 1 + b^2 \breve{c}$ is the upper bound of $\|C\|$ in \eqref{cov-mat-bd}. Using \eqref{eq:1}, \eqref{eq:2} and \eqref{eq:blockmatinv}, we get
\begin{align}
\label{eq:18}
 (X \delta^2 \Lambda X^\T + C)^{-1} X  \beta_{0} &= U
\begin{bmatrix}
  F^{-1} &  - F^{{-1}} \tilde C_{12} \tilde C_{22}^{-1} \\
  - \tilde C_{22}^{{-1}} \tilde C_{12}^{\T} F^{-1}  & \tilde C_{22}^{-1} + \tilde
  C_{22}^{-1} C_{12}^{\T} F^{{-1}} \tilde C_{12} \tilde C_{22}^{-1}
\end{bmatrix}
U^{\T} U
  \begin{bmatrix}
  D \\
  0_{N-p \times p}
  \end{bmatrix}
  V^\T \beta_{0} \nonumber\\
  &= U
\begin{bmatrix}
  F^{-1} D V^{\T} \beta_{0}\\
  - \tilde C_{22}^{{-1}} \tilde C_{12}^{\T} F^{-1}  D V^{\T} \beta_{0}
\end{bmatrix}.
\end{align}
This further implies that the second term on the right-hand side of \eqref{eq:17} satisfies
\begin{align}
\label{eq:19}
 \| (X \delta^2 \Lambda X^\T + C)^{-1} X  \beta_{0} \|_2^{2} &= \|  F^{-1} D V^{\T}
 \beta_{0} \|_2^{2} + \| \tilde C_{22}^{{-1}} \tilde C_{12}^{\T} F^{-1}  D V^{\T}
 \beta_{0}\|_2^{2} \nonumber\\
 &\leq \beta_{0}^{\T} V D F^{-2} D
V^{\T} \beta_{0} + \| \tilde C_{22}^{{-1}} \tilde C_{12}^{\T}\|^2 \|F^{-1}  D V^{\T}
 \beta_{0}\|_2^{2} \nonumber\\
&= \beta_{0}^{\T} V D F^{-2} D
V^{\T} \beta_{0} + \|A(\Gamma)^{\T}\|^2\beta_{0}^{\T} V D F^{-2} D
V^{\T} \beta_{0} \nonumber\\
&\leq \{1 + (1 + b^2 \breve{c})^2\} \beta_{0}^{\T} V D F^{-2} D V^{\T} \beta_{0},
\end{align}
which follows from \eqref{eq:bound_A_Gamma} in Lemma \ref{auxlem5}.

Following the same arguments to derive \eqref{eq:16}, we obtain that
\begin{align}
\label{eq:20}
\beta_{0}^{\T} V D F^{-2} D V^{\T} \beta_{0} &\leq
\frac{1}{\lambda_{\min}^{2}(\tilde D)} \beta_{0}^{\T} V D^{2} V^{\T} \beta_{0} \nonumber \\ &\leq \frac{\lambda_{\max}^{2}(
  D)}{\lambda_{\min}^{2}(\tilde D)} \, \beta_{0}^{\T} V V^{\T} \beta_{0} \nonumber \\
  &=  \frac{\lambda_{\max}^{2}(
  D)}{\lambda_{\min}^{2}(\tilde D)} \,\| \beta_{0}\|_2^{2} \nonumber \\
  & = \kappa^4(X) \| \beta_{0}\|_2^{2}.
\end{align}
Substituting \eqref{eq:19} and \eqref{eq:20} in \eqref{eq:17} gives
\begin{align}\label{bound_T1}
    T_1 &\leq (1 + b^2 \breve{c})^2 \{1 + (1 + b^2 \breve{c})^2\}\kappa^4(X) \| \beta_{0}\|_2^{2} \nonumber \\
    &= \{(1 + b^2 \breve{c})^2 + (1 + b^2 \breve{c})^4 \} \kappa^4(X) \| \beta_{0}\|_2^{2} \nonumber \\
    &\leq \{2(1 + b^4 \breve{c}^2) + 2^3(1 + b^8 \breve{c}^4)\}\kappa^4(X) \| \beta_{0}\|_2^{2},
\end{align}
which follows from the inequality $(u+v)^r \leq 2^{r-1}(u^r + v^r)$. Taking expectations on both sides of \eqref{bound_T1} implies that
\begin{align}
\label{eq:21}
\EE(T_{1}) &\leq \EE_{R,S,Z}\left\{2(1 + b^4 \breve{c}^2) + 2^3(1 + b^8 \breve{c}^4)\right\} \EE_{X} \left\{
\kappa^4(X) \right\} \| \beta_{0}\|_2^{2} \nonumber \\
&\overset{(i)}{\leq} \Bigl[10 + \frac{2^{19}b^4 \sigma_Z^4}{k^4}(q^2 + k^2 + 2)^2\Bigl\{n(q^2+2)+\sum_{i=1}^n m_i^2\Bigr\} \nonumber \\
&\quad + \frac{2^{45}b^8\sigma_Z^8}{k^8}(q^4 + k^4 + 6)^2\Bigl\{n(q^4+6)+\sum_{i=1}^n m_i^4\Bigr\}\Bigr]\EE_{X} \left\{
\kappa^4(X) \right\} \| \beta_{0}\|_2^{2} \nonumber \\
&\overset{(ii)}{\leq} \Bigl[10 + \frac{2^{19}b^4 \sigma_Z^4}{k^4}(4q^4 + 4k^4 + 8)\Bigl\{n(q^2+2)+\sum_{i=1}^n m_i^2\Bigr\} \nonumber \\
&\quad + \frac{2^{45}b^8\sigma_Z^8}{k^8}(4q^8 + 4k^8 + 72)\Bigl\{n(q^4+6)+\sum_{i=1}^n m_i^4\Bigr\}\Bigr]\EE_{X} \left\{
\kappa^4(X) \right\} \| \beta_{0}\|_2^{2} \\
&\leq \Bigl[10 + \frac{2^{19}b^4 \sigma_Z^4}{k^4}(4q^4 + 4k^4 + 8)\Bigl\{n(q^2+2)+n m_{\max}^2\Bigr\} \nonumber \\
&\quad + \frac{2^{45}b^8\sigma_Z^8}{k^8}(4q^8 + 4k^8 + 72)\Bigl\{n(q^4+6)+n m_{\max}^4\Bigr\}\Bigr]\EE_{X} \left\{
\kappa^4(X) \right\} \| \beta_{0}\|_2^{2} \\
&= O\left(\frac{nb^8\sigma_Z^8 q^{12}}{k^8} + \frac{nb^8\sigma_Z^8 m_{\max}^4}{k^8} + nb^8\sigma_Z^8 \right) \EE_{X} \left\{
\kappa^4(X) \right\} \| \beta_{0}\|_2^{2} ,
\end{align}
where $(i)$ holds by putting $r = 4$ and $r = 8$ in \eqref{anrs-bd-2} in Lemma \ref{auxlemma7} and $(ii)$ follows from the inequality $(u+v)^2 \leq 2(u^2 + v^2)$. The lemma is proved.
\end{proof}

\subsection{Auxiliary Results}\label{Section:AuxLemmas}

\begin{lemma}\label{auxlem1}
The two terms in the upperbound for the prediction risk in (\ref{eq:mse}) satisfy
\begin{align}
  \left[ X - X \{  X^{\T} C^{-1} X + (\delta^{2} \Lambda)^{-1}\}^{-1} X^{\T} C^{-1} X \right]
  \beta_0  &= C (X \delta^2 \Lambda X^\T + C)^{-1} X \beta_0,\nonumber\\
  X \{  X^{\T} C^{-1} X + (\delta^{2} \Lambda)^{-1}\}^{-1} X^{\T} C^{-1} \epsilon_0 &= X\delta^2 \Lambda X^\T (X \delta^2 \Lambda X^\T + C)^{-1} \epsilon_0. 
\end{align}
\end{lemma}

\begin{proof}
Apply the Woodbury matrix identity \citep{Har97} to the first term in \eqref{auxlem1} and obtain
\begin{align}\label{eq:first-red}
  [ X - X &\{  X^{\T} C^{-1} X + (\delta^{2} \Lambda)^{-1}\}^{-1} X^{\T} C^{-1} X ]
  \beta_0 \nonumber\\
  &=   [C - X \{  X^{\T} C^{-1} X + (\delta^{2} \Lambda)^{-1}\}^{-1} X^{\T} ] C^{-1} X \beta_0 \nonumber \\
  &= [C -  X \{ \delta^2 \Lambda - \delta^2 \Lambda X^\T (X \delta^2 \Lambda X^\T + C)^{-1} X \delta^2 \Lambda  \} X^{\T}] C^{-1} X \beta_0 \nonumber\\
  &= [C -  X \delta^2 \Lambda X^\T (X \delta^2 \Lambda X^\T + C)^{-1} \{ (X \delta^2 \Lambda X^\T + C) -  X \delta^2 \Lambda  X^{\T}\} ] C^{-1} X \beta_0 \nonumber\\
  &= \{ I -  X \delta^2 \Lambda X^\T (X \delta^2 \Lambda X^\T + C)^{-1}  \} X \beta_0 = C (X \delta^2 \Lambda X^\T + C)^{-1} X \beta_0 .
\end{align}

Following \eqref{eq:first-red}, apply the Woodbury matrix identity again to the
second term in \eqref{auxlem1} and obtain
\begin{align} \label{eq:sec-term-matinv}
  X \{  X^{\T} &C^{-1} X + (\delta^{2} \Lambda)^{-1}\}^{-1} X^{\T} C^{-1} \epsilon_0
   \nonumber\\
   &=
    X \{ \delta^2 \Lambda - \delta^2 \Lambda X^\T (X \delta^2 \Lambda X^\T + C)^{-1} X \delta^2 \Lambda  \} X^{\T} C^{-1} \epsilon_0 \nonumber \\
   &= \{ X\delta^2 \Lambda X^\T - X \delta^2 \Lambda X^\T (X \delta^2 \Lambda X^\T + C)^{-1} X \delta^2 \Lambda X^\T  \} C^{-1} \epsilon_0 \nonumber \\
   &=  X\delta^2 \Lambda X^\T (X \delta^2 \Lambda X^\T + C)^{-1}  \epsilon_{0}.
\end{align}
This completes the proof of the lemma.
\end{proof}

\begin{lemma}\label{auxlem2}
Let $\Gcal = \{\Gamma \in \RR^{k \times k}: \| \Gamma \| \leq b \}$ be a set of matrices. Then, there is a $1/(\log N)$-cover for the metric space $(\Gcal, \| \cdot \|)$ with cardinality $\lceil 2 b k \log N \rceil^{k^2} $.
\end{lemma}

\begin{proof}
  Consider the approximation of $[-b, b]$ by a uniform grid with $\lceil 2b k \log
  N \rceil$ partitions. Every partition in the grid has length less than or equal to $1 / (k \log N)$. If the center of the $l$th partition is $g_l$, then we can form at most $\lceil 2bk \log N \rceil^{k^2}$ $k \times k$  matrices using $\{g_l: l = 1, \ldots, \lceil 2b k \log
  N \rceil\}$. Denote the set formed using these matrices as $\tilde \Gcal = \{\tilde \Gamma_l \in \RR^{k \times k}: l = 1, \ldots, \lceil 2b k \log
  N \rceil^{k^2} \}$.

  We prove the lemma by showing that $\tilde \Gcal$ is a $1/(\log N)$-net for the metric space $(\Gcal, \| \cdot \|)$. For any $\Gamma \in \Gcal$ and $i, j = 1, \ldots, k$,  $ \Gamma_{ij} \in [-b, b]$ because $\| \Gamma \| \leq b$. The construction of $\tilde \Gcal$ implies that there exists a $\tilde \Gamma' \in \tilde \Gcal$ such that  $\mid \Gamma_{ij} -  \tilde \Gamma'_{{ij}} \mid \leq 1 / (k \log N)$ for $i,j = 1, \ldots, k$. Furthermore, summing this inequality over rows and columns gives
  \begin{align}
  \label{eq:25}
  \sum_{j=1}^{k} \mid \Gamma_{ij} - \tilde \Gamma'_{{ij}} \mid \leq \frac{1}{\log
  N}, \quad
  \sum_{i=1}^{k} \mid \Gamma_{ij} - \tilde \Gamma'_{{ij}} \mid \leq \frac{1}{\log
  N}.
  \end{align}
  Using \eqref{eq:25}, in the definitions of $\| \cdot \|_{\infty}$, $\| \cdot \|_1$, and $\| \cdot \|$ norms for matrices implies that
  \begin{align}
  \| \Gamma - \tilde  \Gamma' \|_{\infty} &= \underset{1 \leq i \leq k}{\max}
  \sum_{j=1}^{k} \mid \Gamma_{ij} - \tilde \Gamma'_{{ij}} \mid \leq \frac{1}{\log
  N}, \quad
  \| \Gamma - \tilde \Gamma' \|_{1} = \underset{1 \leq j \leq k}{\max}
  \sum_{i=1}^{k} \mid \Gamma_{ij} - \tilde \Gamma'_{{ij}} \mid \leq \frac{1}{\log
  N}, \nonumber\\
  \|\Gamma - \tilde \Gamma' \| &\leq \left( \| \Gamma - \tilde \Gamma' \|_{1} \| \Gamma - \tilde
  \Gamma' \|_{\infty}\right)^{1/2} \leq \frac{1}{\log N},
  \end{align}
  where the last inequality proves that $\tilde \Gcal$ is a $1 / (\log N)$-cover for the metric space $(\Gcal, \| \cdot \|)$. The lemma is proved.
\end{proof}

\begin{lemma} \label{auxlem5}
Let $\Gcal = \{ \Gamma \in \RR^{k \times k}: \| \Gamma \| \leq b\}$ be the set used in in Lemma \ref{auxlem2} with $\tilde \Gcal$ as its $1/ \log N$-net and $A(\Gamma) = \tilde C_{12}(\Gamma) \tilde C_{22} (\Gamma)^{-1}$ be the matrix that appears in \eqref{eq:13}, where the notation emphasizes the dependence on $\Gamma$. Then, for any $\Gamma \in \Gcal$ and $\tilde \Gamma \in \tilde \Gcal$,
 \begin{align}\label{eq:auxeq1}
     \| A(\Gamma)^{\T} A(\Gamma) - A(\tilde \Gamma)^{\T} A(\tilde \Gamma) \| 
  &\leq 4b (2\breve c +3  b^2 \breve c^2 + b^4 \breve c^3)\, \| \Gamma - \tilde \Gamma\|, \nonumber \\
  \breve c &= \| S \|^2 \, \| R \|^2 \max (\| Z_1 \|^2, \ldots, \|Z_n\|^2 ).
 \end{align}   

\end{lemma}

\begin{proof}
Recall that 
$C_i(\Gamma) = Z_i S^{\T}
\Gamma R R^{\T} \Gamma^{\T} S Z_i^{\T} + I_{m_i}$. For any $\Gamma \in \Gcal$, 
\begin{align}\label{cov-mat-bd}
1 \overset{(i)}{\leq} \lambda_{\min}\{ C_i(\Gamma) \} &\leq \lambda_{\max}\{ C_i(\Gamma) \} \leq 1 + \|Z_i S^\T \Gamma R \|^2 \leq  1 + b^2 \| S \|^2 \| R\|^2 \|Z_i\|^2, \nonumber\\
1 \overset{(ii)}{\leq} \lambda_{\min}\{C(\Gamma)\} &\leq \lambda_{\max}\{C(\Gamma)\} \overset{(iii)}{\leq} 1 + b^2 \| S \|^2 \| R\|^2 \max(\|Z_1\|^2, \ldots, \| Z_n\|^2) \equiv  1 + b^2 \breve c, 
\end{align}
where $(i)$ follows because $C_i(\Gamma) = Z_i S^{\T}
\Gamma R  \left( Z_i S^{\T} \Gamma R \right)^{\T} + I_{m_i}$,  $(ii)$ follows because  $\lambda_{\min}\{C(\Gamma)\} = \min_i \lambda_{\min}\{ C_i(\Gamma) \}$, and $(iii)$ follows because $\lambda_{\max}\{C(\Gamma)\} = \max_i \lambda_{\max}\{ C_i(\Gamma) \}$. For any $\tilde \Gamma \in \tilde \Gcal$ and $\Gamma \in \Gcal$,
\begin{align}\label{lip-c}
   \| C_i(\Gamma) - C_i(\tilde \Gamma) \| &= \| Z_i S^{\T}
\Gamma R R^{\T} \Gamma^{\T} S Z_i^{\T} - Z_i S^{\T}
\tilde \Gamma R R^{\T} \tilde \Gamma^{\T} S Z_i^{\T} \| \nonumber\\
&= \| Z_i S^{\T}
(\Gamma - \tilde \Gamma) R R^{\T} \Gamma^{\T} S Z_i^{\T} + 
Z_i S^{\T} \tilde \Gamma R R^{\T}  (\Gamma^{\T} - \tilde \Gamma^{\T}) S Z_i^{\T} 
  \| \nonumber\\  
&\leq \| Z_i S^\T \| \|R R^{\T} \Gamma^{\T} S Z_i^{\T}\| \|\Gamma - \tilde \Gamma \| + \|Z_i S^{\T} \tilde \Gamma R R^{\T}\| \| S Z_i^{\T} \| \| \Gamma - \tilde \Gamma\| \nonumber\\
&\leq 2 b \, \| Z_i S^\T \| \, \| R R^\T \| \, \| S Z_i^{\T} \| \, \| \Gamma - \tilde \Gamma\| = 
2 b \, \| Z_i \|^2 \| S \|^2 \, \| R \|^2  \, \| \Gamma - \tilde \Gamma\|.
\end{align}
The first term in \eqref{eq:auxeq1} satisfies
\begin{align}\label{eq:lipsch2}
  \| A(\Gamma)^{\T} A(\Gamma) - A(\tilde \Gamma)^{\T} A(\tilde \Gamma) \|  &=  \| A(\Gamma)^{\T} A(\Gamma) - A(\tilde \Gamma)^{\T} A(\Gamma) + A(\tilde \Gamma)^{\T} A(\Gamma) - A(\tilde \Gamma)^{\T} A(\tilde \Gamma) \| \| \nonumber\\
&=  \| \{ A(\Gamma)^{\T}  - A(\tilde \Gamma)^{\T} \} A(\Gamma) + A(\tilde \Gamma)^{\T} \{ A(\Gamma) -  A(\tilde \Gamma)\}\|  \nonumber\\
&\leq \{ \| A(\Gamma) \| + \| A(\tilde \Gamma) \| \} \, \|A(\Gamma)  - A(\tilde \Gamma)\|.
\end{align}
For any $\Gamma \in \Gcal$, 
\begin{align}\label{eq:bound_A_Gamma}
    \| A(\Gamma)\| = \|\tilde C_{12}(\Gamma) \tilde C_{22}(\Gamma)^{-1}\| &\leq \|\tilde C_{12}(\Gamma)\| \|\tilde C_{22}(\Gamma)^{-1}\| = \|\tilde C_{12}(\Gamma)\| \frac{1}{\lambda_{\min} \{ \tilde C_{22}(\Gamma) \} } \nonumber \\
    &\overset{(i)}{\leq}  \frac{\|\tilde C(\Gamma)\|}{\lambda_{\min}\{\tilde C_{22}(\Gamma)\}} \overset{(ii)}{\leq}  \frac{\| C(\Gamma)\|}{\lambda_{\min}\{ C(\Gamma)\}} \nonumber\\
    &\overset{(iii)}{\leq} \| C(\Gamma)\| \overset{(iv)}{\leq}  1 + b^2 \breve c,
\end{align}
where $(i)$ follows from Cauchy's interlacing theorem, $(ii)$ follows from Cauchy's interlacing theorem and the fact that  $\|\tilde C\| = \|U^{\T}CU\| = \|C\|$ because $U$ is an orthonormal matrix, and $(iii)$ and $(iv)$ follow from \eqref{cov-mat-bd}.

Consider the last term in \eqref{eq:lipsch2},
\begin{align}\label{eq:lipsch3}
    \|A(\Gamma)  - A(\tilde \Gamma)\| &= \| A(\Gamma) - \tilde C_{12}(\tilde \Gamma) \tilde C_{22} (\Gamma)^{-1} + \tilde C_{12}(\tilde \Gamma) \tilde C_{22} (\Gamma)^{-1} - A(\tilde \Gamma)\| \nonumber \\
    &= \| \{ \tilde C_{12}(\Gamma) - \tilde C_{12}(\tilde \Gamma) \}\tilde C_{22} (\Gamma)^{-1} + \tilde C_{12}(\tilde \Gamma) \{\tilde C_{22} (\Gamma)^{-1} -\tilde C_{22} (\tilde \Gamma)^{-1} \}\| \nonumber \\
    &\leq \| \tilde C_{22} (\Gamma)^{-1}\| \|\tilde C_{12}(\Gamma) - \tilde C_{12}(\tilde \Gamma)\| + \|\tilde C_{12}(\tilde \Gamma) \| \|\tilde C_{22} (\Gamma)^{-1} -\tilde C_{22} (\tilde \Gamma)^{-1} \| \nonumber \\
    &= \| \tilde C_{22} (\Gamma)^{-1}\| \| \tilde C_{12}(\Gamma) - \tilde C_{12}(\tilde \Gamma)\| + \|\tilde C_{12}(\tilde \Gamma)\| \|\tilde C_{22} (\Gamma)^{-1} \{\tilde C_{22} (\tilde \Gamma) - \tilde C_{22} (\Gamma)\} \tilde C_{22} (\tilde \Gamma)^{-1}\| \nonumber \\
    &\overset{(i)}{\leq} \| \tilde C_{22} (\Gamma)^{-1}\| \|C(\Gamma) - C(\tilde \Gamma)\| + \|\tilde C_{12}(\tilde \Gamma)\| \|\tilde C_{22} (\Gamma)^{-1}\| \|\tilde C_{22} (\tilde \Gamma)^{-1}\|  \| C (\tilde \Gamma) -  C (\Gamma)\| \nonumber \\
     &=\lambda_{\min}^{-1} \{\tilde C_{22} (\Gamma)\} \, ( 1 + \|\tilde C_{12}(\tilde \Gamma)\| \|\tilde C_{22} (\tilde \Gamma)^{-1}\| )  \, \|C(\Gamma) - C(\tilde \Gamma)\|  \nonumber \\
   &\overset{(ii)}{\leq}  \lambda_{\min}^{-1} \{C(\Gamma)\} \, ( 1 + \| C(\Gamma)\| )  \, \|C(\Gamma) - C(\tilde \Gamma)\| \nonumber\\
   &\overset{(iii)}{\leq} (2 + b^2 \breve c) \|C(\Gamma) - C(\tilde \Gamma)\| \nonumber\\
   &\overset{(iv)}\leq  (2 + b^2 \breve c) \, 2 b \,  \| S \|^2 \, \| R \|^2  \, \max(\| Z_1 \|^2, \ldots, \| Z_n \|^2) \, \| \Gamma - \tilde \Gamma\| \nonumber\\
   &= (4b \breve c + 2b^3 \breve c^2) \, \| \Gamma - \tilde \Gamma\|,
\end{align}
where $(i)$ holds because $\tilde C_{12}(\Gamma) - \tilde C_{12}(\tilde \Gamma)$ and $\tilde C_{22}(\Gamma) - \tilde C_{22}(\tilde \Gamma)$ are $p \times (N-p)$ and $(N-p) \times (N-p)$ submatrices of $\tilde C(\Gamma) - \tilde C(\tilde \Gamma)$, $(ii)$ and $(iii)$ follow from the arguments used in \eqref{eq:bound_A_Gamma}, and $(iv)$ follows from taking maximum over the bound in \eqref{lip-c}. Combining \eqref{eq:lipsch1}, \eqref{eq:lipsch2}, 
\eqref{eq:lipsch3}, and \eqref{eq:bound_A_Gamma} implies that 
\begin{align}\label{eq:lipsch4}
  \| A(\Gamma)^{\T} A(\Gamma) - A(\tilde \Gamma)^{\T} A(\tilde \Gamma) \| &\leq 2 (1 + b^2 \breve c) (4 b \breve c + 2 b^3 \breve c^2)\, \| \Gamma - \tilde \Gamma\| \nonumber\\
  &= 
  4b (2\breve c +3  b^2 \breve c^2 + b^4 \breve c^3)\, \| \Gamma - \tilde \Gamma\| .
\end{align}
The lemma is proved.
\end{proof}

\begin{lemma}\label{auxlemma7}
Let $R$, $S$, $Z_1, \ldots, Z_n$ be the matrices used to define the CME model in equation (2) of the main manuscript. Then,
    \begin{align} \label{anrs-bd}
       &\EE(\| R \|_2^2)= \EE(\| S \|_2^2) \leq \frac{4}{k} (q + k + 2), \quad \EE(\| Z_i \|_2^2) \leq 4\sigma_Z^2 (q + m_i + 2), \quad i = 1, \ldots, n,   \nonumber\\   
       &\EE \max (\| Z_1 \|_2^2, \ldots, \| Z_n \|_2^2) \leq  4 \sigma_Z^2 \, \{n(q + 2) + N\}, \nonumber\\
       &\EE \left\{\| S \|_2^2 \, \| R \|_2^2 \, \max (\| Z_1 \|_2^2, \ldots, \| Z_n \|_2^2) \right\} \leq 64 \; \frac{\sigma_Z^2}{k^2} (q + k + 2)^2 \{n(q+2) + N\},
    \end{align}
    and for $r \geq 3$,
    \begin{align}\label{anrs-bd-2}
         &\EE \| S \|_2^r = \EE \| R \|_2^r \leq \frac{2^{2r - 2}}{k^{r/2}}  \{q^{r/2} + k^{r/2} + 2^{1-r/2}r\Gamma(r/2) \}, \nonumber\\
      &\EE \| Z_i \|_2^r \leq 2^{2r - 2}\sigma _Z^{r} \, \{q^{r/2} + m_i^{r/2} + 2^{1-r/2}r\Gamma(r/2) \},\nonumber\\
      &\EE \left\{\| S \|_2^r \, \| R \|_2^r \, \max (\| Z_1 \|_2^r, \ldots, \| Z_n \|_2^r) \right\} \nonumber \\
     &\leq 2^{6r-6}\frac{\sigma_Z^r}{k^r} \; \{q^{r/2} + k^{r/2} + 2^{1-r/2}r\Gamma(r/2) \}^2 \left[ nq^{r/2} + n2^{1-r/2}r\Gamma(r/2) + \sum_{i=1}^n m_i^{r/2} \right].
    \end{align}
\end{lemma}

\begin{proof}
Let $A$ be a $k \times q$ matrix whose entries are independent and identically distributed as $\Ncal(0,1)$. Then,
  \begin{align}\label{eq:exp1}
      \EE \|  A \|_2^2 = \int_0^{\infty} \PP (\|  A \|_2^2 > t) \, dt.
  \end{align}
  Based on Corollary 5.35 in \citet{Ver12}, we decompose the last integral as
  \begin{align}\label{eq:sbd1}
      \int_0^{\infty} \PP (\|  A \|_2^2 > t) \, dt &= 
      \int_0^{(\sqrt{2q} + \sqrt{2k})^2} \PP (\|  A \|_2^2 > t) \, dt + 
      \int_{(\sqrt{2q} + \sqrt{2k})^2}^{\infty} \PP (\|  A \|_2^2 > t) \, dt \nonumber \\
      &\leq (\sqrt{2q} + \sqrt{2k})^2 + \int_{(\sqrt{2q} + \sqrt{2k})^2}^{\infty} \PP (\|  A \|_2^2 > t) \, dt \nonumber \\
      &= (\sqrt{2q} + \sqrt{2k})^2 + \int_{0}^{\infty} \PP \{ \|  A \|_2^2 > u + (\sqrt{2q} + \sqrt{2k})^2 \} \, du \nonumber \\
      &\overset{(i)}{\leq} 2 (2q + 2k) + \int_{0}^{\infty} \PP \left\{ \|  A \|_2^2 > \left( \sqrt{\frac{u}{2}} + \frac{\sqrt{2q} + \sqrt{2k}}{\sqrt{2}} \right)^2 \right\} \, du,
  \end{align}
  where $(i)$ follows from using the inequality $(u + v)^2 \leq 2 (u^2 + v^2)$ in two terms. Corollary 5.35 in \citet{Ver12} implies that
  \begin{align*}
      \PP \left\{ \|  A \|_2^2 > \left( \sqrt{\frac{u}{2}} + \frac{\sqrt{2q} + \sqrt{2k}}{\sqrt{2}} \right)^2 \right\} = 
      \PP \left( \|  A \|_2 >  \sqrt{\frac{u}{2}} + \sqrt{q} + \sqrt{k} \right) \leq 2 e^{-\frac{u}{4}}.
  \end{align*}
  Using this bound in \eqref{eq:sbd1} gives
  \begin{align}\label{exp-bd}
      \EE \|  A \|_2^2 = \int_0^{\infty} \PP (\|  A \|_2^2 > t) \, dt \leq 4 (q + k) + 2 \int_{0}^{\infty} e^{-\frac{u}{4}}\, du = 4 (q + k) + 8.
  \end{align}

  We now use the bound in \eqref{exp-bd} to derive the bounds for $\EE \| S\|^2$, $\EE \| R \|^2$, and $\EE \|Z_i \|^2$. Let $R = k^{-1/2} \tilde R$, $\tilde S = k^{-1/2} S$, and $Z_i = \sigma_Z \tilde Z_i $, where $\tilde R$, $\tilde S$, and $\tilde Z_i$ entries are independent and identically distributed as $\Ncal(0, 1)$. Then,  \eqref{exp-bd} implies that 
  \begin{align*}
      \EE \| S \|_2^2 &=\EE \| k^{-1/2} \tilde S  \|_2^2 = \EE \| R \|_2^2 =\EE \| k^{-1/2} \tilde R  \|_2^2 \leq \frac{4}{k} (q+k+2), \\
      \EE \| Z_i \|_2^2 &=\EE \| \sigma_Z \tilde Z_i  \|_2^2 \leq 4 \sigma_Z^2 (q + m_i + 2).
  \end{align*}
  Finally, union bound and independence of $R, S, Z_1, \ldots, Z_n$ imply that
  \begin{align*}
      &\EE \max (\| Z_1 \|_2^2, \ldots, \| Z_n \|_2^2) \leq \sum_{i=1}^n \EE \| Z_i \|_2^2 \leq  4 \sigma_Z^2 \{n(q+2) + N\}, \\
     &\EE \left\{\| S \|_2^2 \, \| R \|_2^2 \, \max (\| Z_1 \|_2^2, \ldots, \| Z_n \|_2^2) \right\} \leq 64 \sigma_Z^2 \; \frac{(q+k+2)^2}{k^2}  \{n(q+2) + N\}.
  \end{align*}

Consider the extension for these arguments for an $r \geq 3$,
\begin{align}\label{eq:genpower1}
    \EE\|A\|_2^r &= \int_0^{\infty} \PP (\|  A \|_2^r > t) \, dt \nonumber \\
    &= \int_0^{\left(\sqrt{2^{2-2/r}q} + \sqrt{2^{2-2/r}k} \right)^r} \PP (\|  A \|_2^r > t) \, dt + 
      \int_{\left(\sqrt{2^{2-2/r}q} + \sqrt{2^{2-2/r}k} \right)^r}^{\infty} \PP (\|  A \|_2^r > t) \, dt \nonumber \\
    &\leq \left(\sqrt{2^{2-2/r}q} + \sqrt{2^{2-2/r}k} \right)^r + \int_{\left(\sqrt{2^{2-2/r}q} + \sqrt{2^{2-2/r}k} \right)^r}^{\infty} \PP (\|  A \|_2^r > t) \, dt \nonumber \\
    &= \left(\sqrt{2^{2-2/r}q} + \sqrt{2^{2-2/r}k} \right)^r + \int_{0}^{\infty} \PP \left\{ \|  A \|_2^r > u + \left(\sqrt{2^{2-2/r}q} + \sqrt{2^{2-2/r}k} \right)^r \right\} \, du \nonumber \\
    &= \left(\sqrt{2^{2-2/r}q} + \sqrt{2^{2-2/r}k} \right)^r + \int_{0}^{\infty} \PP \left[ \|  A \|_2^r > 2^{r-1}\left\{\left(\frac{u^{1/r}}{2^{1-1/r}}\right)^r + \left(\frac{\sqrt{2^{2-2/r}q} + \sqrt{2^{2-2/r}k}}{2^{1-1/r}} \right)^r\right\}\right]\, du \nonumber \\
    &\overset{(i)}{\leq} 2^{r-1}\left(2^{r-1}q^{r/2} + 2^{r-1}k^{r/2} \right) + \int_{0}^{\infty} \PP \left\{ \|  A \|_2^r > \left(\frac{u^{1/r}}{2^{1-1/r}} + \frac{\sqrt{2^{2-2/r}q} + \sqrt{2^{2-2/r}k}}{2^{1-1/r}} \right)^r\right\} \, du \nonumber\\
\end{align}
where $(i)$ follows from using the inequality $(u + v)^r \leq 2^{r-1} (u^r + v^r)$ in both the terms. Corollary 5.35 in \cite{Ver12} implies that
\begin{align*}
    \PP \left\{ \|  A \|_2^r > \left(\frac{u^{1/r}}{2^{1-1/r}} + \frac{\sqrt{2^{2-2/r}q} + \sqrt{2^{2-2/r}k}}{2^{1-1/r}} \right)^r\right\} &= \PP \left\{ \|  A \|_2 > \left(\frac{u^{1/r}}{2^{1-1/r}} + \frac{\sqrt{2^{2-2/r}q} + \sqrt{2^{2-2/r}k}}{2^{1-1/r}} \right)\right\} \nonumber\\
    &= \PP \left\{ \|  A \|_2 > \left(\frac{u^{1/r}}{2^{1-1/r}} + \sqrt{q} + \sqrt{k}\right)\right\}\nonumber \\
    &\leq 2e^{-\frac{u^{2/r}}{2^{3-2/r}}},
\end{align*}
which we use in \eqref{eq:genpower1} to obtain
\begin{align}\label{eq:genpower2}
   \EE \|  A \|_2^r &= \int_0^{\infty} \PP (\|  A \|_2^r > t) \, dt \nonumber\\ &\leq 2^{2r - 2}(q^{r/2} + k^{r/2}) + 2 \int_{0}^{\infty} e^{-\frac{u^{2/r}}{2^{3-2/r}}}\, du \nonumber\\
   &= 2^{2r - 2}(q^{r/2} + k^{r/2}) + r\left\{2^{(3r)/2 - 1} \right\}\Gamma(r/2) \nonumber \\
   &= 2^{2r - 2} \{q^{r/2} + k^{r/2} + 2^{1-r/2}r\Gamma(r/2) \}. 
\end{align}
Now, let $R = k^{-1/2} \tilde R$, $\tilde S = k^{-1/2} S$, and $Z_i = \sigma_Z \tilde Z_i $, where the entries of $\tilde R$, $\tilde S$, and $\tilde Z_i$ are independently and identically distributed as $\Ncal(0, 1)$. Using \eqref{eq:genpower2} then implies
  \begin{align*}
      \EE \| S \|_2^r &=\EE \| k^{-1/2} \tilde S  \|_2^r = \EE \| R \|_2^r =\EE \| k^{-1/2} \tilde R  \|_2^r \leq \frac{1}{k^{r/2}} 2^{2r - 2} \{q^{r/2} + k^{r/2} + 2^{1-r/2}r\Gamma(r/2) \}, \\
      \EE \| Z_i \|_2^r &=\EE \| \sigma_Z \tilde Z_i  \|_2^r \leq \sigma_Z^{r} 2^{2r - 2} \{q^{r/2} + m_i^{r/2} + 2^{1-r/2}r\Gamma(r/2) \}.
  \end{align*}
Finally, using the independence of $R, S, Z_1, \dots, Z_n$ and union bound, we get
  \begin{align*}
      &\EE \{\max (\| Z_1 \|_2^r, \ldots, \| Z_n \|_2^r)\} \leq \sum_{i=1}^n \EE \| Z_i \|_2^r \leq \sigma_Z^{r}2^{2r - 2} \left[ nq^{r/2} + n2^{1-r/2}r\Gamma(r/2) + \sum_{i=1}^n m_i^{r/2} \right], \\
     &\EE \left\{\| S \|_2^r \, \| R \|_2^r \, \max (\| Z_1 \|_2^r, \ldots, \| Z_n \|_2^r) \right\} \nonumber \\
     &\leq 2^{6r-6}\frac{\sigma_Z^r}{k^r} \; \{q^{r/2} + k^{r/2} + 2^{1-r/2}r\Gamma(r/2) \}^2 \left[ nq^{r/2} + n2^{1-r/2}r\Gamma(r/2) + \sum_{i=1}^n m_i^{r/2} \right].
  \end{align*}
    The lemma is proved.

\end{proof}

\section{Posterior Computation}\label{Section:PostComp}
This section presents the detailed derivations of the full conditional distributions presented in Algorithm 1 of the main manuscript. First, consider the derivation of the posterior updates for the compressed parameter $\gamma$. Let $\check{y} = y - X\beta$ denote the residual after removing the fixed effects, and $\check{Z}$ be the design matrix associated with the compressed random effects; see \eqref{eq:gen-lik-2} in the main manuscript. The conjugate Gaussian prior for the conditional Gaussian likelihood implies that the full conditional density of $\gamma$ takes the form 
\begin{align}
    f(\gamma \mid y, \beta, \tau^2) &\propto f(\check y \mid \gamma, \beta, \tau^2) f(\gamma) \nonumber \\
    &\propto (\tau^2)^{-\frac{N}{2}} \exp \left\{ -\frac{1}{2\tau^2} (\check y - \check Z \gamma)^{\T} (\check y - \check Z \gamma) \right\} (\sigma_{\gamma}^2)^{-\frac{k_1 k_2}{2}} \exp \left( -\frac{1}{2\sigma_\gamma^2} \gamma^{\T} \gamma\right) \nonumber \\
    &\propto (\sigma_{\gamma}^2)^{-\frac{k_1 k_2}{2}} \exp \left[ -\frac{1}{2} \left\{ \gamma^{\T} \left( \frac{1}{\tau^2}\check Z^{\T}\check Z + \frac{1}{\sigma_\gamma^2} I_{k_1 k_2}\right)\gamma - \frac{2}{\tau^2}\check y^{\T} \check Z \gamma \right\}\right] \nonumber \\
    &\propto \exp \left\{ -\frac{1}{2}(\gamma - \mu_\gamma)^{\T} \Sigma_\gamma^{-1}(\gamma - \mu_\gamma)\right\}, \nonumber \\
    & \equiv \mathcal{N}(\mu_\gamma, \Sigma_\gamma),
\end{align}
where $\mu_\gamma = \tau^{-2} \Sigma_{\gamma}\check Z^{\T}\check y$, and $\Sigma_\gamma = \left( \tau^{-2} \check Z^{\T}\check Z + \sigma_\gamma^{-2} I_{k_1 k_2} \right)^{-1}$ represent the mean and covariance parameters of the Gaussian density.

We now derive the full conditional distribution of $(\beta, \tau^2)$, marginalizing over the imputed compressed random effects $d_i$'s; see model \eqref{eq:model-cycle2} in the main manuscript. Using the scale mixture representation of the half-Cauchy distribution assigned a priori to the global and local shrinkage parameters, we express the hierarchy in the Horseshoe prior as
\begin{align}
    \beta_j \mid \lambda_j^2, \delta^2, \tau^2 &\overset{ind}{\sim} \mathcal{N}(0, \lambda_j^2 \delta^2 \tau^2), \nonumber \\
    \lambda_j^2 \mid \nu_j &\overset{ind}{\sim} \mathcal{IG}(1/2, 1/\nu_j), \ j = 1, \dots, p, \nonumber \\
    \delta^2 \mid \xi &\sim \mathcal{IG}(1/2, 1/\xi), \nonumber \\
    \tau^2 &\sim \mathcal{IG}(a_0, b_0), \nonumber \\
    \nu_1, \dots, \nu_p, \xi &\overset{ind}{\sim} \mathcal{IG}(1/2, 1), 
\end{align}
where $\mathcal{IG}(a, b)$ denotes the inverse-gamma distribution with shape parameter $a$ and scale parameter $b$. The full conditional density of the fixed effect parameter $\beta$ is
\begin{align}
    f(\beta \mid \tau^2, \lambda_1, \dots, \lambda_p, \delta, \nu_1, \dots, \nu_p, \xi, \gamma) &\propto (\tau^2)^{-\frac{N}{2}} \exp \left\{-\frac{1}{2\tau^2}(y^* - X^* \beta)^{\T}(y^* - X^* \beta) \right\} \nonumber\\
    &\quad \, \exp\left\{ -\frac{1}{2\tau^2}\beta^{\T}(\delta^2 \Lambda)^{-1}\beta\right\} \nonumber \\
    &\propto \exp \left\{ -\frac{1}{2}(\beta - \mu_\beta)^{\T} \Sigma_\beta^{-1}(\beta - \mu_\beta)\right\} \nonumber \\
    & \equiv \mathcal{N}(\mu_\beta, \tau^2\Sigma_\beta),
\end{align}
where $\Sigma_\beta = \left\{X^{*\T}X^{*} + (\delta^2 \Lambda)^{-1} \right\}^{-1}$, $\mu_\beta = \Sigma_{\beta}X^{*\T} y^{*}$, and $\Lambda = \text{diag}(\lambda_1^2, \dots, \lambda_p^2)$.  Given $\beta$, we leverage the conjugacy of the inverse-gamma prior to derive the full conditional density of the error variance $\tau^2$ as
\begin{align}
    f(\tau^2 \mid \beta, \lambda_1, \dots, \lambda_p, \delta, \nu_1, \dots, \nu_p, \xi, \gamma) &\propto (\tau^2)^{-\frac{N}{2}} \exp \left\{-\frac{1}{2\tau^2}(y^* - X^* \beta)^{\T}(y^* - X^* \beta) \right\} \nonumber \\
    &(\tau^2)^{-\frac{p}{2}}\exp\left\{ -\frac{1}{2\tau^2}\beta^{\T}(\delta^2 \Lambda)^{-1}\beta\right\} (\tau^2)^{-a_0 - 1} exp\left( -\frac{b_0}{\tau^2}\right) \nonumber \\
    &\propto (\tau^2)^{-\left(a_0 + \frac{N+p}{2}\right)-1} \exp \left[- \frac{1}{\tau^2}\left\{ b_0 + \frac{r^{*\T}r^*}{2} + \frac{\beta^{\T}(\delta^2 \Lambda)^{-1}\beta}{2}\right\} \right] \nonumber \\
    &\equiv \mathcal{IG}\left(a_0 + \frac{N+p}{2}, b_0 + \frac{r^{*\T}r^*}{2} + \frac{\beta^{\T}(\delta^2 \Lambda)^{-1}\beta}{2}\right),
\end{align}
where $r^* = y^* - X^* \beta$.

The full conditional density of the squared global shrinkage parameter $\delta^2$ satisfies
\begin{align}
    f(\delta^2 \mid \beta, \tau^2, \lambda_1, \dots, \lambda_p, \nu_1, \dots, \nu_p, \xi, \gamma) &\propto \mid \delta^2 \Lambda \mid^{-\frac{1}{2}} \exp\left\{ -\frac{1}{2\tau^2}\beta^{\T}(\delta^2 \Lambda)^{-1}\beta\right\} (\delta^2)^{-\frac{1}{2} - 1} \exp\left(-\frac{1/\xi}{\delta^2} \right) \nonumber \\
    &\propto (\delta^2)^{-\frac{p+1}{2} - 1} \exp\left\{ -\frac{1}{\delta^2} \left( \frac{1}{\xi} + \frac{1}{2\tau^2}\sum_{j=1}^p \frac{\beta_j^2}{\lambda_j^2}\right)\right\} \nonumber \\
    &\equiv \mathcal{IG}\left(\frac{p+1}{2}, \frac{1}{\xi} + \frac{1}{2\tau^2}\sum_{j=1}^p \frac{\beta_j^2}{\lambda_j^2}\right).
\end{align}
The full conditional density of the squared local shrinkage parameter $\lambda_j^2$ is
\begin{align}
    f(\lambda_j^2 \mid \beta, \tau^2, \delta, \nu_1, \dots, \nu_p, \xi, \gamma) &\propto \mid \delta^2 \Lambda \mid^{-\frac{1}{2}} \exp\left\{  -\frac{1}{2\tau^2}\beta^{\T}(\delta^2 \Lambda)^{-1}\beta\right\} (\lambda_j^2)^{-\frac{1}{2} - 1} \exp\left(-\frac{1/\nu_j}{\lambda_j^2} \right) \nonumber \\
    &\propto (\lambda_j^2)^{-\frac{1}{2}}\exp\left(-\frac{\beta_j^2}{2\delta^2\tau^2\lambda_j^2}\right) (\lambda_j^2)^{-\frac{1}{2} - 1} \exp\left(-\frac{1/\nu_j}{\lambda_j^2} \right) \nonumber \\
    &\equiv \mathcal{IG}\left(1, \frac{1}{\nu_j} + \frac{\beta_j^2}{2\delta^2\tau^2} \right).
\end{align}
Finally, the full conditional density of the auxiliary variables, $\xi, \nu_1, \ldots, \nu_p$, are
\begin{align}
    f(\xi \mid \beta, \tau^2, \lambda_1, \dots, \lambda_p, \delta, \nu_1, \dots, \nu_p, \gamma) &\propto (\xi)^{-\frac{1}{2}}(\delta^2)^{-\frac{1}{2} - 1} \exp\left(-\frac{1/\xi}{\delta^2} \right) (\xi)^{-\frac{1}{2} - 1}\exp\left( -\frac{1}{\xi}\right)\nonumber \\
    &\propto (\xi)^{-1 -1} \exp \left\{ -\frac{1}{\xi} \left( 1 + \frac{1}{\delta^2}\right)\right\} \nonumber \\
    &\equiv \mathcal{IG}\left(1, 1 + \frac{1}{\delta^2} \right),
\end{align}
and
\begin{align}
    f(\nu_j \mid \beta, \tau^2, \lambda_1, \dots, \lambda_p, \delta, \xi, \gamma) &\propto (\nu_j)^{-\frac{1}{2}}(\lambda_j^2)^{-\frac{1}{2} - 1} \exp\left(-\frac{1/\nu_j}{\lambda_j^2} \right) (\nu_j)^{-\frac{1}{2} - 1}\exp\left( -\frac{1}{\nu_j}\right)\nonumber \\
    &\propto (\nu_j)^{-1 -1} \exp \left\{ -\frac{1}{\nu_j} \left( 1 + \frac{1}{\lambda_j^2}\right)\right\} \nonumber \\
    &\equiv \mathcal{IG}\left(1, 1 + \frac{1}{\lambda_j^2} \right).
\end{align}

\section{Additional Experimental Results}\label{Section:ExtraSim}

We conduct an additional simulation to evaluate the robustness of the proposed method. In this setup, the fixed effect covariate for each sample is generated from a multivariate normal distribution with a Toeplitz covariance matrix, introducing correlation among the covariates. Specifically, each row of $X \in \mathbb{R}^{N \times p}$ is generated independently from $\Ncal(0, \Sigma^X)$, where $\Sigma^X_{jj'} = 0.5 ^ {|j - j'|}$ for any $j, j' \in \{1, \ldots, p\}$. In this correlated design setting, we observe that CME with compressed covariance dimensions $k_1 = k_2 = 3$ maintains its competitive performance across coverage, interval width, fixed effects selection, and predictive accuracy metrics. These results show that CME's performance generalizes beyond the independent case.

CME has the best fixed effects recovery performance and matches the oracle benchmark across a range of true covariance structures of the random effects. For any $\Sigma$, CME with small compression dimensions, such as $k_1 = k_2 = 3$, yields the narrowest credible intervals for various signal strengths in $\beta$ attaining the nominal coverage of 95\% (Figures \ref{fig:VS_covg_allSigma_ToepX} and \ref{fig:VS_width_allSigma_ToepX}). Although CME's coverage is below the nominal level when the signal is weakest ($\beta = 0.05$), its empirical coverage improves as the sample size increases with cluster size $m$. In contrast, both the penalized quasi-likelihood methods exhibit consistent undercoverage of the non-zero coefficients in $\beta$ when the true $\Sigma$ is low-rank. When $\Sigma$ is full-rank, PQL-1 and CME have comparable coverage because PQL-1 uses a non-singular identity matrix as the proxy matrix. While CME's performance is largely insensitive to $k_2$, increasing $k_1$ leads to narrower credible intervals when the sample size is larger ($m = 8$ and $m = 12$). Setting both the compression dimensions to be large ($k_1 = k_2 = 14$), however, consistently results in undercoverage of the signals across all sample sizes. We conclude that when the compression dimensions are small, CME combined with the $S_2M$ algorithm outperforms its competitors in fixed effects selection accuracy as the sample size increases with $m$ (Table \ref{Table_tprfpr_allSigma_ToepX}).

CME with $k_1 = k_2 = 3$ maintains its superior predictive accuracy compared to the quasi-likelihood competitors across various true covariance structures (Figure \ref{fig:rel_mspe_allSigma_ToepX}). As the sample size increases with $m$, the mean square prediction error of CME decreases, indicating its gain in accuracy. For the smaller compression dimensions, CME produces the shortest prediction intervals while attaining the nominal 95\% coverage (Table \ref{Table_covg&relWidth_PI_CME_allSigma_ToepX}).  No penalized likelihood competitors are included in Table \ref{Table_covg&relWidth_PI_CME_allSigma_ToepX} because the variance component estimation fails due to the assumption that $q < m$ when $m_1 = \cdots = m_n$, which is violated in high-dimensional settings; see Section 4.1 in \cite{Lietal21} for the detailed conditions.

In summary, CME outperforms its competitors in fixed effects selection and prediction when the compression dimensions remain small ($k_1 = k_2 = 3$), with further improvements in both aspects as the sample size increases with cluster size $m$. These findings closely align with those observed in the independent covariate setting discussed in Section \ref{Section:Sim_results} of the main manuscript, highlighting the practical utility of CME and its robustness to varying covariance structures of the random effects and correlations among the fixed effect covariates.

\begin{figure}[htbp]
  \centering
  \begin{subfigure}[b]{0.78\textwidth}
    \includegraphics[width=\textwidth]{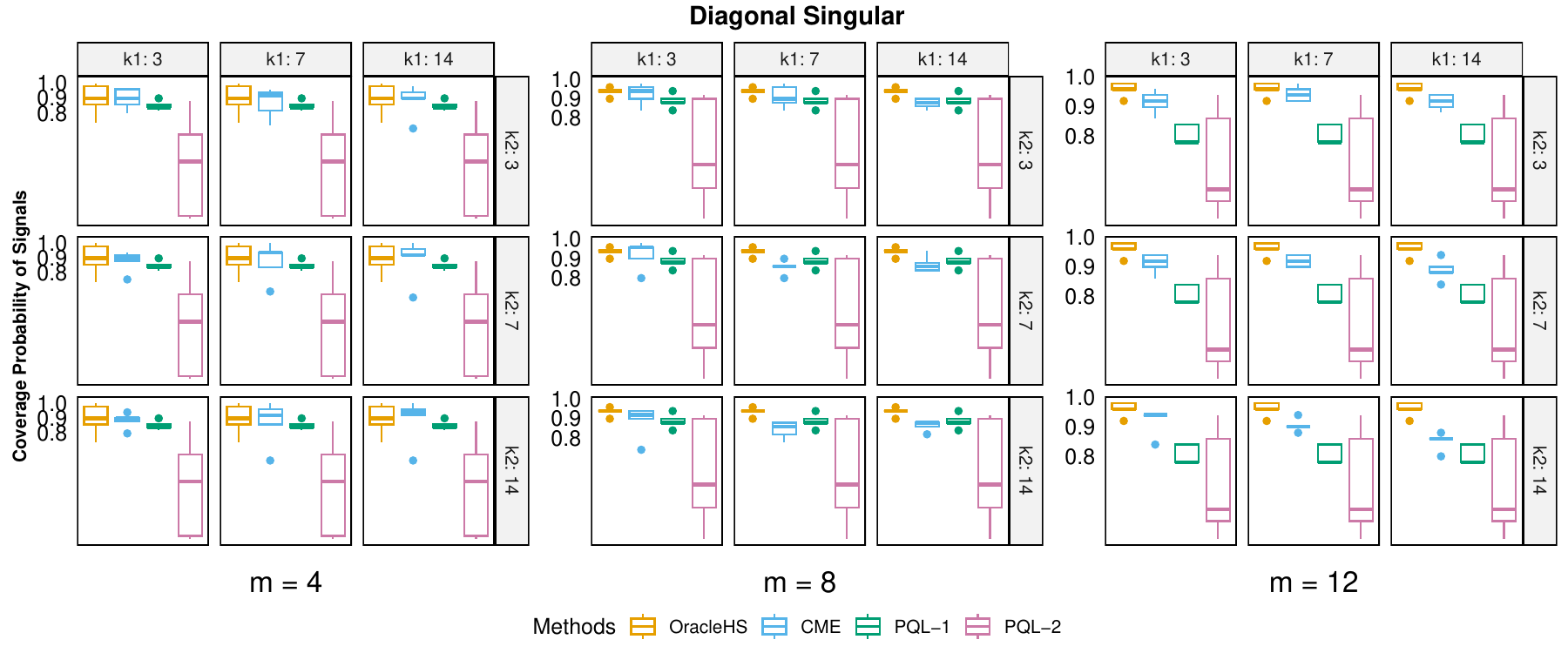}
    \label{fig:subfigA_covg}
  \end{subfigure}
  \\
  \begin{subfigure}[b]{0.78\textwidth}
    \includegraphics[width=\textwidth]{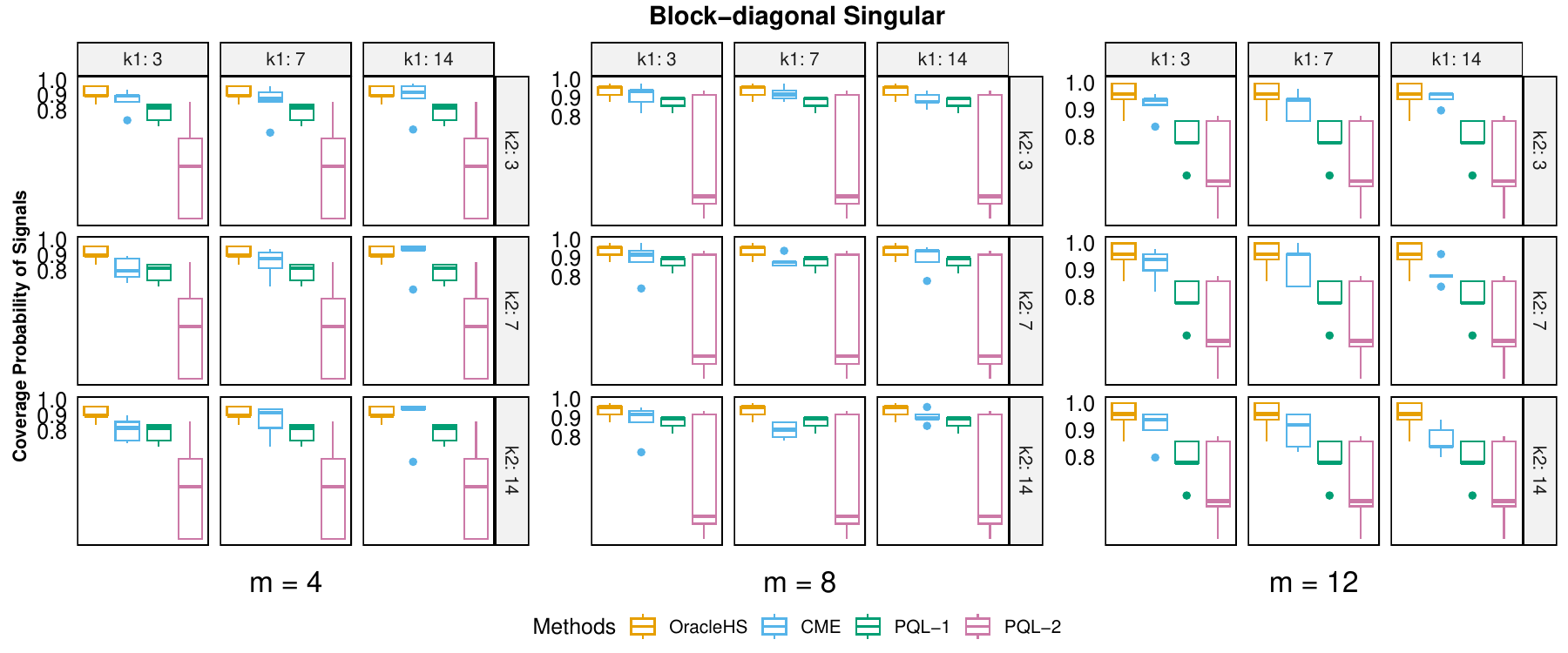}
    \label{fig:subfigB_covg}
  \end{subfigure}
  \\
  \begin{subfigure}[b]{0.78\textwidth}
    \includegraphics[width=\textwidth]{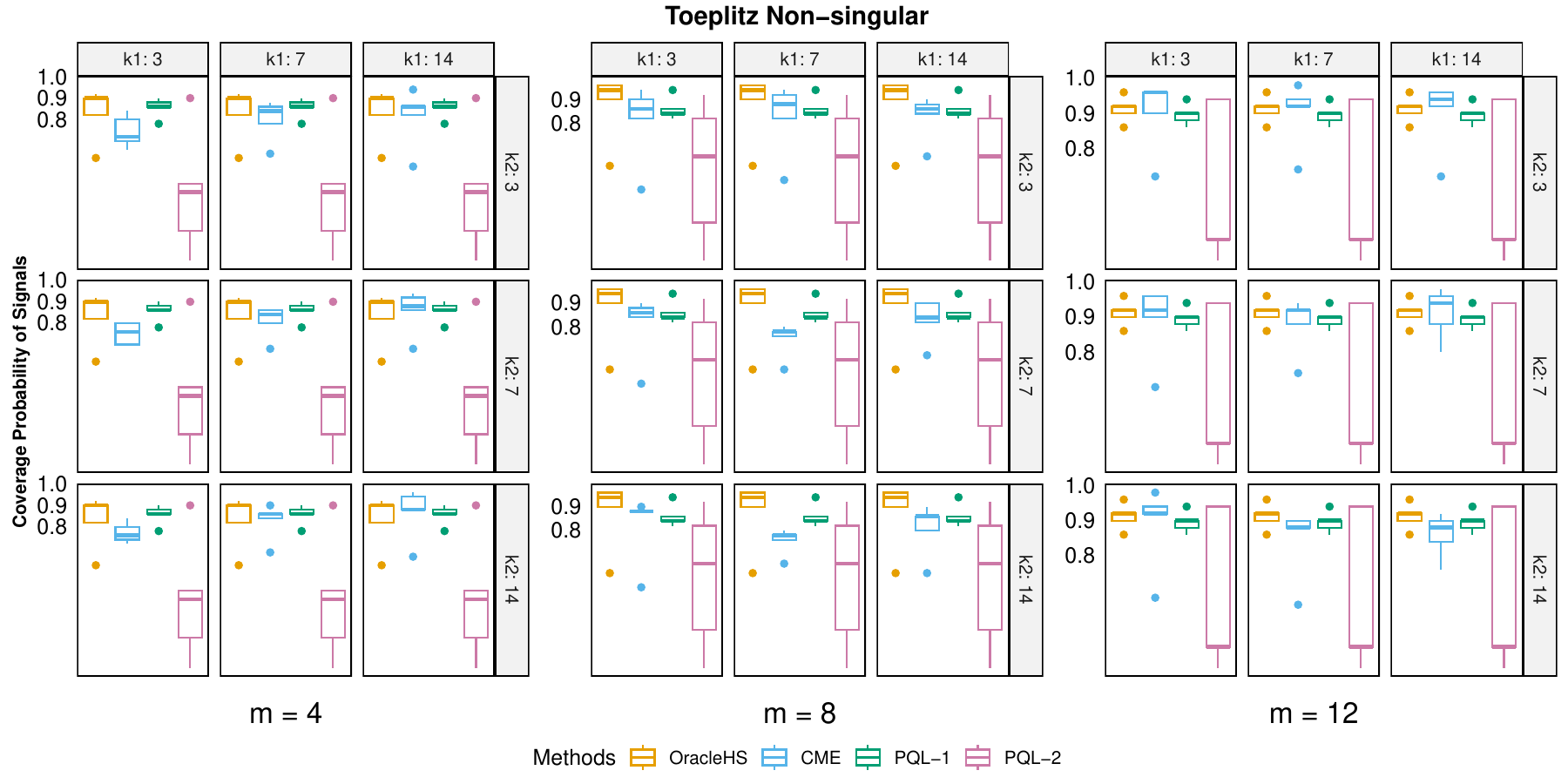}
    \label{fig:subfigC_covg}
  \end{subfigure}
  \caption{Coverage of signals depending on cluster size $m \in \{4, 8, 12\}$ and compression dimensions ($k_1 \in \{3, 7, 14\}$ and $k_2 \in \{3, 7, 14 \}$). The random effects covariance matrices are diagonal (singular), block-diagonal (singular), and Toeplitz (non-singular).   
  CME outperforms its regularized quasi-likelihood competitors in the presence of correlated fixed effects. OracleHS, oracle with Horseshoe prior; CME, compressed mixed-effects model; PQL-1, the penalized quasi-likelihood approach of \citet{FanLi12}; PQL-2, the penalized quasi-likelihood approach of \citet{Lietal21}.}
  \label{fig:VS_covg_allSigma_ToepX}
\end{figure}

\begin{figure}[htbp]
  \centering
  \begin{subfigure}[b]{0.78\textwidth}
    \includegraphics[width=\textwidth]{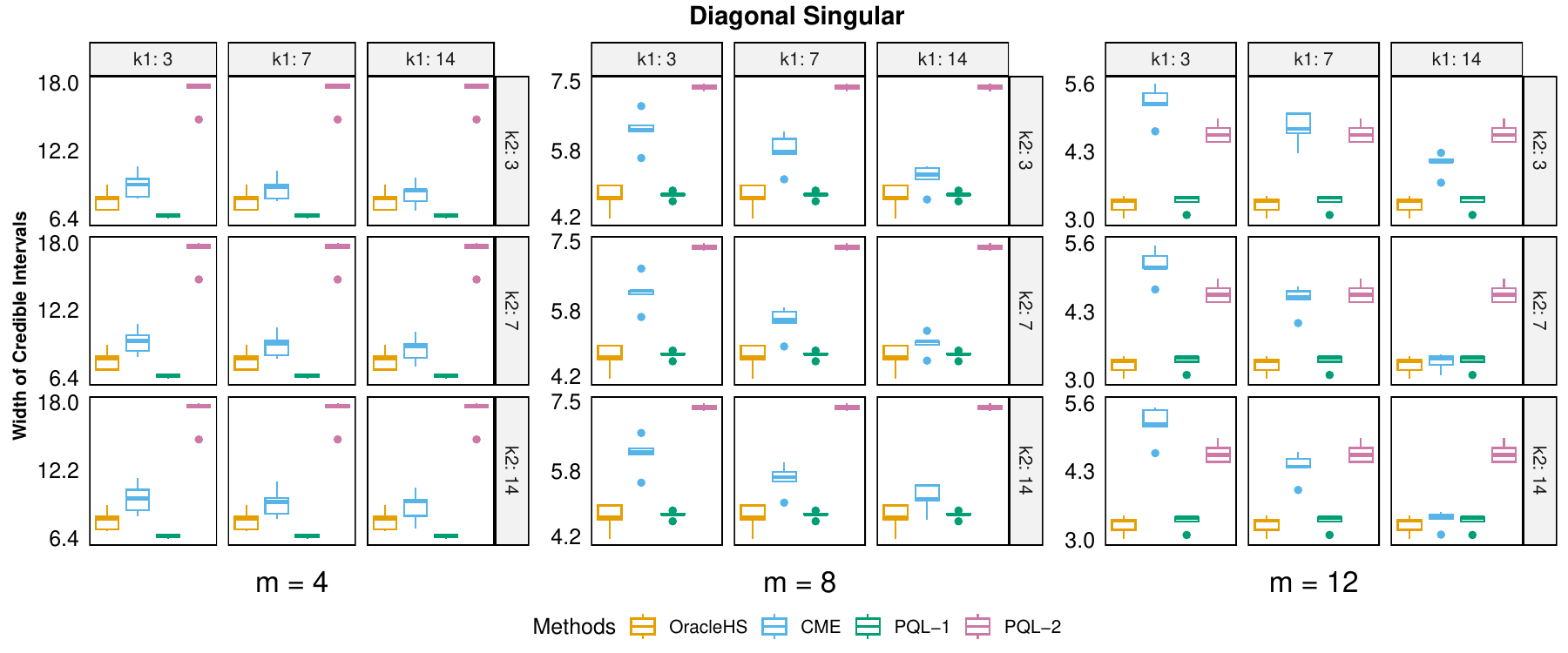}
    \label{fig:subfigA_width}
  \end{subfigure}
  \\
  \begin{subfigure}[b]{0.78\textwidth}
    \includegraphics[width=\textwidth]{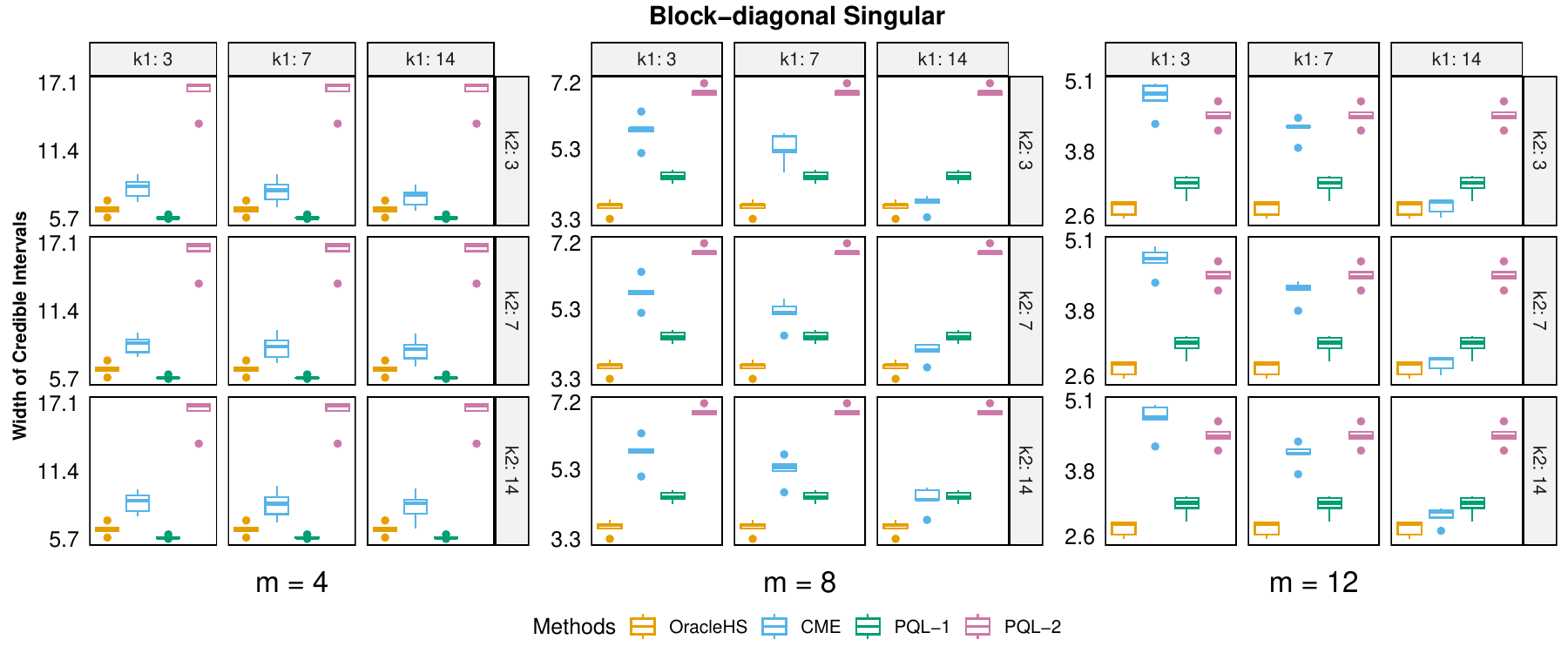}
    \label{fig:subfigB_width}
  \end{subfigure}
  \\
  \begin{subfigure}[b]{0.78\textwidth}
    \includegraphics[width=\textwidth]{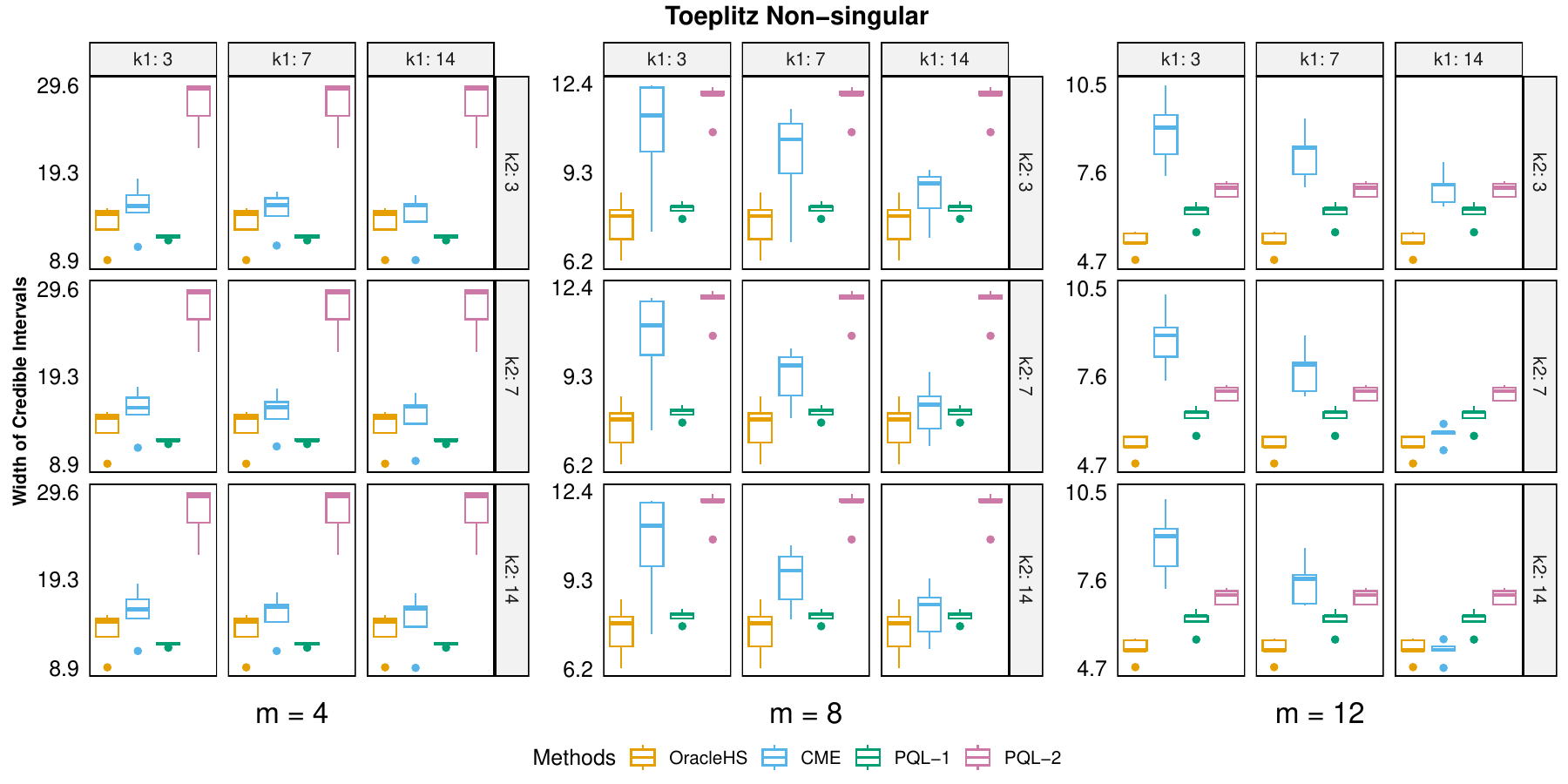}
  \end{subfigure}
  \caption{Width of credible intervals of signals depending on cluster size $m \in \{4, 8, 12\}$ and compression dimensions ($k_1 \in \{3, 7, 14\}$ and $k_2 \in \{3, 7, 14 \}$). The random effects covariance matrices are diagonal (singular), block-diagonal (singular), and Toeplitz (non-singular).   
  CME outperforms its regularized quasi-likelihood competitors in the presence of correlated fixed effects. OracleHS, oracle with Horseshoe prior; CME, compressed mixed-effects model; PQL-1, the penalized quasi-likelihood approach of \citet{FanLi12}; PQL-2, the penalized quasi-likelihood approach of \citet{Lietal21}.}
  \label{fig:VS_width_allSigma_ToepX}
\end{figure}

\begin{table}[h!]
\centering
\caption{True positive rates (TPR) and false positive rates (FPR in parentheses) in fixed effects selection averaged over 50 replicates in the presence of correlated fixed effects, depending on cluster size $m \in \{4, 8, 12\}$ and random effects covariance matrices as diagonal positive semi-definite, block-diagonal positive semi-definite, and Toeplitz positive definite. OracleHS, oracle with Horseshoe prior; CME, compressed mixed-effects model; PQL-1, the penalized quasi-likelihood approach of \citet{FanLi12}; PQL-2, the penalized quasi-likelihood approach of \citet{Lietal21}. Results for CME are presented for compression dimensions $k_1 = k_2 = 3$.}
\scalebox{0.70}{
\begin{tabular}{|c|c|ccc|}
\hline
\multirow{2}{*}{$\Sigma$} & \multirow{2}{*}{Methods} & \multicolumn{3}{c|} {TPR (FPR)}\\ \cline{3-5}
& & $m = 4$ & $m = 8$ & $m = 12$\\
\hline
\multirow{6}{*}{Diagonal psd} & & & & \\
& CME & 0.91 (0.16) & 0.96 (0.26) & 0.99 (0.31)\\
&  OracleHS & 0.98 (0.49) & 0.99 (0.46) & 1 (0.43)\\
& PQL-1 & 0.98 (0.06) & 1 (0.09) & 1 (0.32)\\
& PQL-2 & 0.67 (0) & 0.87 (0) & 0.98 (0)\\
& & & & \\
\hline
\multirow{6}{*}{Block-diagonal psd} & & & & \\
& CME & 0.90 (0.14) & 0.98 (0.31) & 0.99 (0.33)\\ 
 & OracleHS & 0.98 (0.44) & 1 (0.47) & 1 (0.41)\\
 & PQL-1 & 0.98 (0.07) & 1 (0.09) & 1 (0.35)\\
 & PQL-2 & 0.67 (0) & 0.87 (0) & 1(0)\\
  & & & & \\
 \hline
 \multirow{6}{*}{Toeplitz pd} & & & & \\
& CME & 0.68 (0.10) & 0.92 (0.30) & 0.95 (0.35)\\ 
 & OracleHS & 0.97 (0.56) & 0.98 (0.47) & 1 (0.48)\\
& PQL-1 & 0.92 (0.05) & 0.98 (0.07) & 1 (0.16)\\
 & PQL-2 & 0.58 (0) & 0.8 (0) & 0.94 (0)\\
 & & & & \\
\hline
 \end{tabular}
}
\label{Table_tprfpr_allSigma_ToepX}
\end{table}

\begin{figure}[h!]
    \centering
    \includegraphics[scale = 0.50]{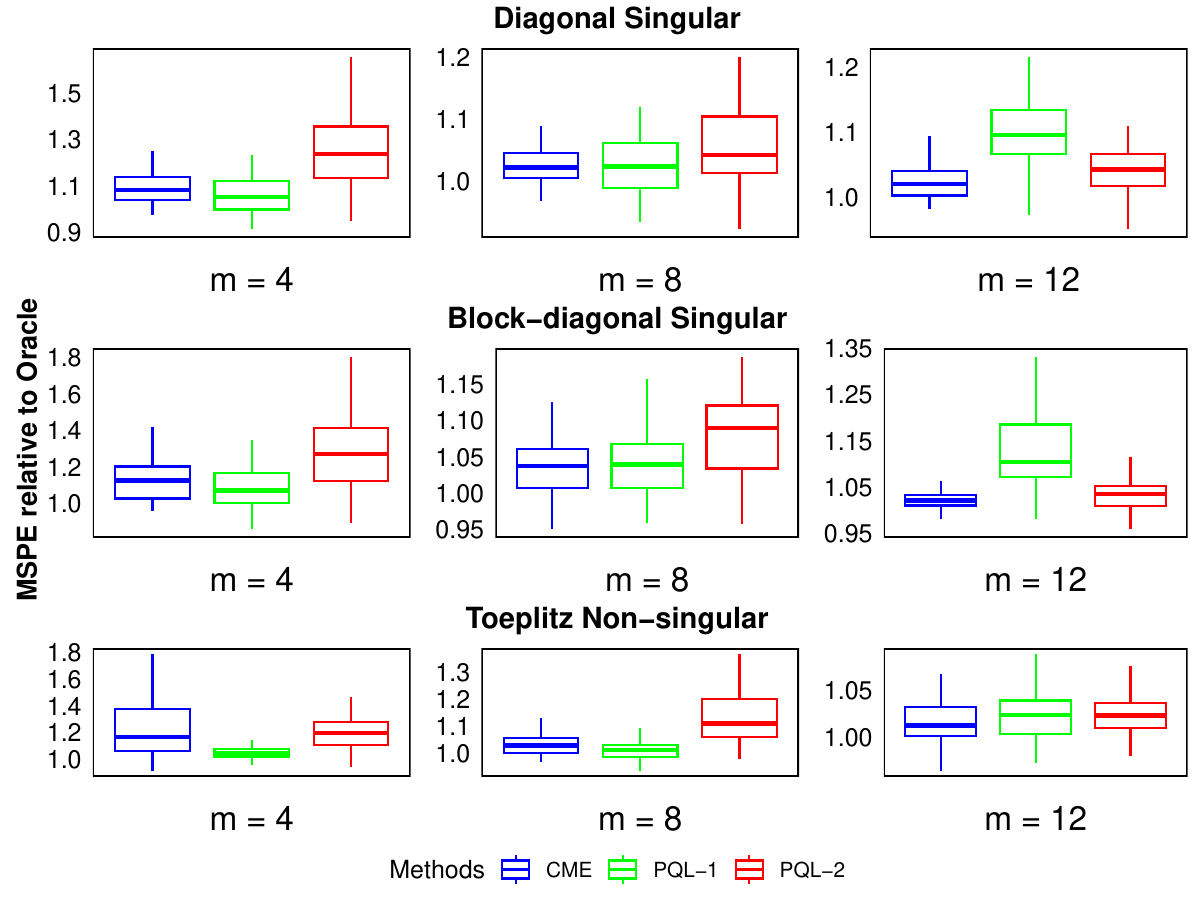}
    \caption{Relative mean square prediction error (MSPE) averaged over 50 replications in the presence of correlated fixed effects. Relative MSPE is computed by dividing the MSPE for a method by that of OracleHS, the oracle with Horseshoe prior. CME, compressed mixed-effects model; PQL-1, the penalized quasi-likelihood approach of \citet{FanLi12}; PQL-2, the penalized quasi-likelihood approach of \citet{Lietal21}. Results for CME are presented for the best choice of compression dimensions $k_1 = k_2 = 3$.}
    \label{fig:rel_mspe_allSigma_ToepX}
\end{figure}

\begin{table}[h!]
\centering
\caption{Coverage of 95\% Prediction Intervals produced by the CME model in the presence of correlated fixed effects with their relative width reported in parentheses. Relative width is computed by dividing the width of a prediction interval produced by CME model by that of the OracleHS, the oracle with Horseshoe prior. All metrics are averaged over 50 replications.}
\resizebox{\columnwidth}{!}{
\begin{tabular}{|c|cccc|ccc|ccc|}
\hline
\multirow{3}{*}{$\Sigma$} & \multirow{3}{*}{$k_1$} & \multicolumn{9}{c|} {$k_2$}\\ \cline{2-11}
  & & \multicolumn{3}{c|} {3} & \multicolumn{3}{c|} {7} & \multicolumn{3}{c|} {14}\\ \cline{2-11}
& & $m = 4$ & $m = 8$ & $m = 12$ & $m = 4$ & $m = 8$ & $m = 12$ & $m = 4$ & $m = 8$ & $m = 12$\\
\hline
\multirow{5}{*}{Diagonal psd} & & & & & & & & & & \\
 & 3 & \textbf{0.95 (1.14)} & \textbf{0.95 (1.07)} & \textbf{0.95 (1.06)} & 0.95 (1.35) & 0.96 (1.14) & 0.96 (1.11) & 0.96 (1.42) & 0.97 (1.30) & 0.97 (1.22) \\
& 7 & 0.97 (1.34) & 0.96 (1.19) & 0.96 (1.14) & 0.97 (1.38) & 0.97 (1.42) & 0.98 (1.28) & 0.98 (1.40) & 0.97 (1.44) & 0.98 (1.46) \\
& 14 & 0.98 (1.32) & 0.97 (1.36) & 0.98 (1.30) & 0.98 (1.32) & 0.97 (1.30) & 0.99 (1.34) & 0.99 (1.36) & 0.98 (1.32) & 0.98 (1.34) \\
& & & & & & & & & &\\
\hline
\multirow{5}{*}{Block-diagonal psd} & & & & & & & & & & \\
 & 3 & \textbf{0.94 (1.19)} & \textbf{0.95 (1.10)} & \textbf{0.95 (1.09)} & 0.95 (1.40) & 0.95 (1.18) & 0.95 (1.14) & 0.96 (1.48) & 0.97 (1.34) & 0.96 (1.25) \\ 
& 7 & 0.96 (1.36) & 0.96 (1.21) & 0.96 (1.17) & 0.97 (1.42) & 0.97 (1.45) & 0.97 (1.29) & 0.98 (1.44) & 0.97 (1.51) & 0.98 (1.49) \\
 & 14 & 0.97 (1.33) & 0.98 (1.30) & 0.98 (1.27) & 0.98 (1.34) & 0.97 (1.26) & 0.98 (1.31) & 0.98 (1.40) & 0.98 (1.30) & 0.98 (1.31) \\
  & & & & & & & & & &\\
 \hline
 \multirow{5}{*}{Toeplitz pd} & & & & & & & & & & \\
 & 3 & \textbf{0.94 (1.17)} & \textbf{0.95 (1.08)} & \textbf{0.95 (1.06)} & 0.94 (1.29) & 0.96 (1.16) & 0.96 (1.11) & 0.95 (1.30) & 0.97 (1.31) & 0.97 (1.22) \\ 
& 7 & 0.96 (1.26) & 0.96 (1.20) & 0.97 (1.14) & 0.97 (1.27) & 0.96 (1.38) & 0.98 (1.28) & 0.97 (1.26) & 0.96 (1.36) & 0.99 (1.46)\\
 & 14 & 0.97 (1.21) & 0.97 (1.33) & 0.98 (1.29) & 0.98 (1.19) & 0.97 (1.22) & 0.98 (1.28) & 0.98 (1.20) & 0.97 (1.21) & 0.98 (1.23) \\
 & & & & & & & & & &\\
\hline
 \end{tabular}
}
\label{Table_covg&relWidth_PI_CME_allSigma_ToepX}
\end{table}

\bibliographystyle{Chicago}
\bibliography{papers}

\end{document}